\documentclass[journal]{IEEEtran}
\usepackage{amsmath}    
\usepackage{graphics,epsfig,subfigure}    
\usepackage{verbatim}   
\usepackage{color}      

\usepackage{subfigure}  
\usepackage{hyperref}   
\usepackage{amssymb}
\usepackage{epstopdf}
\usepackage{graphicx}
\graphicspath{{/Figures/}}
\usepackage{wasysym}
\usepackage{wrapfig}
\usepackage{sidecap}
\usepackage{float}
\usepackage{subfigure}
\usepackage{enumerate}
\usepackage{graphicx}
\graphicspath{{./figures/}}
\usepackage{algorithmic}
\usepackage{algorithm}
\usepackage{amsthm}
\usepackage{multirow}

\newtheorem{lemma}{\textbf{Lemma}}
\newtheorem{theorem}{\textbf{Theorem}}
\newtheorem{corollary}{\textbf{Corollary}}

\begin{document}
%
\title{Opportunistic Routing for the Vehicular Energy Network}

\author{Albert~Y.S.~Lam and
        Victor~O.K.~Li
\thanks{A.Y.S. Lam is with the Department
of Computer Science, Hong Kong Baptist University, Kowloon Tong, Hong Kong(e-mail: albertlam@ieee.org).}
\thanks{V.O.K. Li is with the Department
of Electrical and Electronic Engineering, The University of Hong Kong, Pokfulam Road, Hong Kong (e-mail: vli@eee.hku.hk).}
}

%


\maketitle

\begin{abstract}
Vehicular energy network (VEN) is a vehicular network which can transport energy over a large geographical area by means of electric vehicles (EVs). In the near future, an abundance of EVs, plentiful generation of the renewables, and  mature wireless energy transfer and vehicular communication technologies will expedite the realization of VEN. 
To transmit energy from a source to a destination, we need to establish energy paths, which are composed of segments of vehicular routes, while satisfying various design objectives. In this paper, we develop a method to construct all energy paths for a particular energy source-destination pair, followed by some analytical results of the method. We describe how to utilize the energy paths to develop optimization models for different design goals and propose two  solutions. We also develop a heuristic for the power loss minimization problem. We compare the performance of the three solution methods with artificial and real-world traffic networks and provide a comprehensive comparison in terms of solution quality, computation time, solvable problem size, and applicability. This paper lays the foundations of VEN routing.

\end{abstract}
\begin{IEEEkeywords}
Electric vehicle, energy path, optimization, routing, vehicular energy network.
\end{IEEEkeywords}

\IEEEpeerreviewmaketitle

\section{Introduction}

\IEEEPARstart{T}{he} smart grid \cite{smartgridintro} is expected to be  more reliable with the support of fault detection and self-healing. Its network topology is more flexible allowing bi-directional energy flows and distributed generation. Power transmission can be more efficient with less redundancy in power lines and higher utilization of generators. It can sustain greater penetration of renewable energy to support various nations' energy mandates. It can also enable new markets by accommodating different operational strategies.
These features of smart grid nurture many novel research ideas and business opportunities. 
However, since the operating power grid requires very high reliability and security, the power grid operators, e.g., regional transmission organizations and independent system operators, would hesitate to integrate the smart grid innovations into the existing power grids. One of the possible ways to bring them to the table is to let the existing power system function as it is and to let the new innovations operate in an environment coupling with the existing system in a loose and flexible manner.
To do this, we need a supplementary power delivery system which is easy to build and can complement the existing power system.

We consider the following phenomena and technological trends:
\begin{itemize}
\item To fight global warming and climate change, many nations have engaged to reduce carbon footprint and to promote alternative energy sources like the renewables. According to \cite{REN21}, many nations have very high renewable acquisitions and renewable energy capacity has grown tremendously.   
\item Electric vehicles (EVs) represent more efficient and greener form of conveyance which can run on existing road infrastructure. Together with the Electric Vehicles Initiative \cite{EVI}, it is expected that there will be tremendous number of EVs in the transportation system in the near future. 
\item The road network is one of the most well-established public infrastructure covering most regions of the World involving civil activities.  
\item Wireless power transfer technologies allow power to be transferred over an air gap. Many companies and research institutes have been actively improving the dynamic charging technologies for EVs \cite{qualcomm,stanford}. They not only facilitate EV charging on the move but also help lower the EV market price by allowing EV battery with smaller size. 
\item Vehicular ad-hoc network (VANET) is a mature technology facilitating vehicle-to-vehicle and vehicle-to-roadside, and vehicle-to-infrastructure communications. The vehicular network allows us to acquire various EV status information. 
\end{itemize}
All these facilitate the realization of vehicular energy network (VEN) \cite{VEN}. VEN is a vehicular network capable of transmitting energy effectively over a large geographical area by means of EVs. It is built upon the existing  road networks, where EVs traverse the network along certain routes completely based on the drivers' wills.  Wireless (dis)charging facilities and small energy storage are installed at certain road junctions. With VANET, we can acquire the travel plans of the participating EVs and a set of vehicular routes with known traffic flows can be determined. When a particular EV comes across a road junction, it is  wirelessly charged with (a small amount of) energy. When it reaches an appropriate road junction, we discharge the energy from the EV, such that the energy is then stored in the storage facility such as a battery for subsequently charging another EV. Hence we utilize EVs as carriers to convey energy from one place to another. With proper selection of EVs to carry energy, the energy transmission rates over the road connections are highly controllable. 

Renewables like solar and wind energies are usually generated in remote locations where the absence of electric transmission systems prevents the generated energy from being brought back to the main grid. Even if such transmission systems exist, we may sometimes  disconnect  them intentionally to avoid uncontrollable situations due to intermittency of renewables. 
As road networks are usually available, we can construct the corresponding overlay VENs. As the energy transmission rates of VEN are controllable, the transmission schemes can be made adaptable to the intermittency. Hence VEN is particularly suitable for promoting the use of renewables.
Even if the participating EVs carry very little amount of energy each time, VEN has been shown to be effective at conveying a substantial amount of energy across a large geographical region in a short period of time \cite{VEN}.


The characteristics of VEN can be summarized as follows. Energy is carried by EVs in the form of  ``energy packets'' and this packet switching-like design makes the energy transmission scheme of VEN highly controllable. We do not require the EVs to actively participate in the sense that  they neither need to follow any instructions  to deviate from their own paths nor slow down for (dis)charging purposes. VEN is very flexible; VEN can be easily built on top of any existing road network and the energy source and destination can be altered freely without physically modifying the infrastructure. VEN incurs very low capital cost as most of the required equipment is off-the-shelf.\footnote{Dynamic charging is primarily designed for the ease of EV charging. We just adopt this technology for the purpose of conveying energy on VEN.}

The rest of the paper is organized as follows. In Section \ref{sec:related}, we review related work on EVs and their developments in the smart grid context. Section \ref{sec:model} describes the VEN system model and analytically quantifies the system variables and their relationships. In Section \ref{sec:wholeset}, we propose a method to construct energy paths, followed by some analytical results. We also discuss the utilization of the energy paths and derive two general solutions for solving VEN problems of various design objectives. Section \ref{sec:heuristic} introduces a heuristic for the power loss minimization problem. We evaluate the performance of the three proposed solution methods
in Section \ref{sec:performance} and conclude in Section \ref{sec:conclusion}.


\section{Related Work} \label{sec:related}

EVs take a very important role in energy management in the smart grid. When compared to the capacity of the power grid, the capacity of an EV is very small. However, an aggregation of many EVs can become a huge load or power source. An energy market can be set up to trade energy between aggregations of EVs with the main grid in a vehicle-to-grid system \cite{V2GMarket}. EVs can also be used to provide regulation services to the power system in a distributed fashion \cite{regulation}. In practice, charging stations are currently the main source of energy supply to EVs and their locations can affect the mobility pattern of vehicles \cite{EVCPP}. With VEN, EVs are used to transport energy across an area,  complementing the power network. EVs can also obtain energy to support mobility from VEN. We can see that VEN brings a new dimension of functionality in the smart grid.

VEN is specially designed for conveying energy while VANET aims to disseminate information.  Yet they both utilize the vehicular network to provide additional services over geographical areas other than transportation of  passengers or goods. They share many similarities on the underlying routing principle making use of the opportunistic contacts of vehicles for energy or data exchanges. 
\cite{geopps} proposed an opportunistic routing protocol for VANET by exploiting vehicular mobility patterns and geographical information provided in navigation systems.
\cite{Greedy} focused on position-based routing with topological knowledge for VANET in a city environment.
\cite{Move} proposed an opportunistic forwarding scheme, which utilizes velocity information to make forwarding decisions.
However, routing algorithms developed for VANET may not be applicable to VEN as data and energy are different in nature. Data packets are different from one another, i.e., we are dealing with a multi-commodity routing problem, although they can be replicated to increase the chance of transmission success. However, ``energy packets'' are indistinguishable, i.e., we have a a single commodity routing problem, and we cannot replicate energy.

Mobile electrical grid (or called EV energy network) proposed in \cite{EVnetICC} has a similar but different design as VEN. It does make use of EVs for energy transmission and distribution but it requires the involved EVs to actively participate in the energy transmission process by stopping at particular locations for charging and discharging. However, with dynamic   (dis)charging technologies, VEN can function transparently to the EV drivers. In \cite{VEN}, we provided an extensive analytical framework for further performance study of VEN.

\cite{EVnetMASS} discussed routing in the mobile electrical grid in the presence of traffic congestion by assuming every route capable of transmitting unlimited amount of energy.
It constructed energy routes heuristically in terms of shortest paths. 
\cite{EVnetISGT2014} relaxed the above unlimited energy assumption and considered a simple flow model for multiple route construction.
However, the shortest-path strategy may not be appropriate when the focus is not on energy loss. Even so, we will show that this strategy may not give the optimal results. In this paper, we provide the fundamentals of VEN routing which can be applied to problems of different system objectives.


\section{System Model} \label{sec:model}

We follow \cite{VEN} to define VEN. VEN is built upon a vehicular network, where EVs traverse different locations through some vehicular routes. We first define the underlay vehicular network and then the overlay energy network.

\subsection{Vehicular Network} 
Suppose that there is a fleet of EVs, which participates in VEN, traversing the vehicular network. We model the network with a directed graph $G(\mathcal{N},\mathcal{A})$, where $\mathcal{N}$ is the set of road junctions and $\mathcal{A}$ is the set of road segments or arcs connecting the road junctions. For each arc $a=(tail(a),head(a))\in\mathcal{A}$, EVs go along $a$ from $tail(a)\in\mathcal{N}$ to $head(a)\in\mathcal{N}$.
A vehicular route is a sequence of physically connected arcs and the $i$-th route is denoted by $r_i=\langle a^i_1,\ldots,a^i_{|r_i|} \rangle$, which is composed of $|r_i|$ arcs. $r_i$ is known to the system if there exists some traffic of EVs starting at $tail(a^i_1)$ and ending at $head(a^i_{|r_i|})$.
Without loss of generality, we assume all vehicular routes are loop-free. It is generally true that an EV will not pass through a repeated road junction along a single vehicular route in normal situations. Even if it does, we can consider a looped route as two different routes.
Consider the example shown in Fig. \ref{fig:loop} in which a vehicular route $r$  is composed of four arcs, i.e., $r=\langle a_1, a_2, a_3, a_4\rangle$. If we remove the loop formed by $a_2$ and $a_3$, the vehicular flow along $r$ will be broken due to a skip of time spent on the loop. However, without the loop, we can consider $r$ as two separate routes, as $r_1=\langle a_1 \rangle$ and $r_2=\langle a_4 \rangle$. In this way, we can consider any looped vehicular route as two independent unlooped routes.

\begin{figure}[!t]
\centering
\includegraphics[width=3.2in]{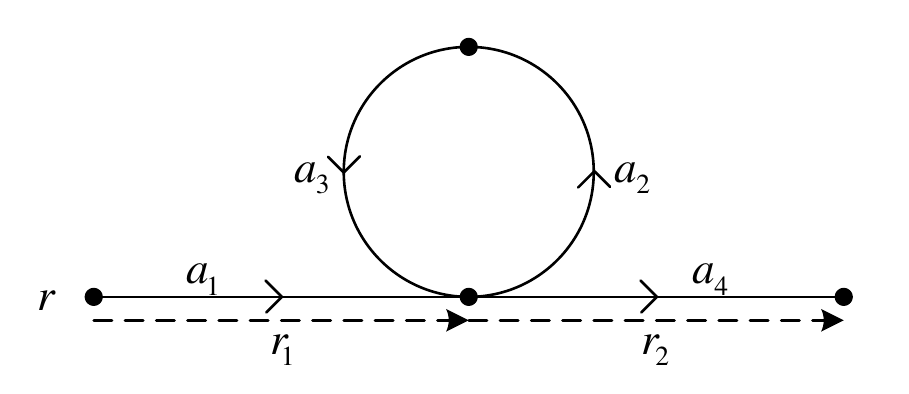}
\caption{Transformation of a vehicular route with a loop.}
\label{fig:loop}
\end{figure}

The $n$-th arc of $r_i$ is denoted as $r_i(n)$, i.e., $r_i(n)=a^i_n$. With $n<m$, we also define the sub-route of $r_i$ connecting $r_i(n)$ and $r_i(m)$ as $r_i(n,m)$, i.e., $r_i(n,m) = \langle a^i_n,\ldots,a^i_m \rangle$.
With the support of VANET and other communication technologies, most vehicles are connected in the future \cite{connected}. It is possible to track the current locations of the participating EVs through certain sensing technologies, e.g., through GPS. However, the availability of their subsequent moves also depends on the nature of the EVs and the degree of information disclosure. For example, suppose that there is an EV which intends to go along the route $\langle a_1, a_2,a_3,a_4 \rangle$. If the EV fully discloses its travel plan (e.g., it is a public transport), the route $\langle a_1, a_2,a_3,a_4 \rangle$ will be recorded in the system. On the other hand, the EV may not disclose any information at all and its location will be tracked when it passes through $tail(a_1)$. Once it is on $a_1$, it must go along $a_1$ and pass through $head(a_1)$ for sure. Thus we can still have a route $\langle a_1 \rangle$ recorded in the system. Depending on the willingness of the driver, it may also result in $\langle a_1, a_2 \rangle$ or $\langle a_1, a_2,a_3 \rangle$. In this way, we make up a set of vehicular route $\mathcal{R}$ and the traffic flow $f_i$ on $r_i\in \mathcal{R}$ can be estimated by counting the number of EVs traversing $r_i$ for a certain period of time.

\subsection{Energy Network} 

Assume that each road junction in $G(\mathcal{N},\mathcal{A})$ is equipped with wireless energy transfer equipment and a small energy storage. In this way, when an EV passes through such a road junction, it may discharge some of its energy with the wireless discharging device and the energy will be stored. When another EV passes by the road junction, the previously stored energy can be transferred wirelessly to it. Due to the advancement  in dynamic (dis)charging technology, the charging and discharging processes will be transparent to the EV driver; the EVs do not need to stop or even slow down in order to complete the charging or discharging process. Moreover, as shown in \cite{VEN}, the amount of energy needed to be transferred in each dynamic (dis)charging is very small, and dynamic (dis)charging can be considered instantaneous.

We define a set of energy sources $\mathcal{N}_s\subset\mathcal{N}$ and a set of energy destinations $\mathcal{N}_d\subset\mathcal{N}$. Each source $s\in \mathcal{N}_s$ has a source of energy connected, which can be a renewable energy source (e.g., solar parks and wind farms), a big energy storage, or even a traditional power generator. An energy destination $t \in \mathcal{N}_t$ has a load attached. 

\subsubsection{Energy Path}
$s$ is connected to $t$ through a set of energy paths $\mathcal{P}(s,t)$. Each energy path $p_j(s,t)\in \mathcal{P}(s,t)$ is composed of  segments of vehicular routes, i.e., $p_j(s,t) = \langle r_1^j(n_1,m_1),\ldots, r_i^j(n_i,m_i),\ldots,r_{|p_j|}^j(n_{|p_j|},m_{|p_j|}) \rangle$, where $r_i^j(n_i,m_i)$ is the $i$-th segment of $p_j(s,t)$ and it is also the sub-route of $r_i^j$ starting at its $n_i$-th arc and ending at its $m_i$-th arc. $|p_j|$ is  the number of vehicular sub-routes adopted to construct $p_j(s,t)$. $p_j(s,t)$ is deemed valid if the following conditions are satisfied:
\begin{enumerate}
	\item[(i)] $tail(r_1^j(n_1))=s$;
	\item[(ii)] $head(r_i^j(m_i)) = tail(r_{i+1}^j(n_{i+1}))$, for $i=1,\ldots,$ \mbox{$|p_j|-1$}; and
	\item[(iii)] $head(r_{|p_j|}^j(m_{|p_j|}))=t$.
\end{enumerate}
From the energy perspective, when some energy is transmitted along $p_j(s,t)$, some EVs are first wirelessly charged at $tail(r_1(n_1))$ and go along $r_1(n_1,m_1)$ with the energy. At $head(r_1(m_1))$ (i.e., $tail(r_2(n_2))$, the EVs are discharged and the energy is stored in the storage. Next the energy is drawn from the storage and used to charge other EVs along $r_2(n_2,m_2)$. This process continues along $p_j(s,t)$ until the energy reaches $t$, i.e., $head(r_{|p_j|}(m_{|p_j|}))$. Strictly speaking, the sequence of charging and discharging events should follow the time order; at a road junction, energy needs to be first discharged from an EV before it can be used to charge another EV. However, energy does not have an identity and one unit of energy from one source is identical to one unit from another source. In this sense, at a road junction along $p_j(s,t)$, charging can take place before discharging. In other words, energy which has possibly come from another source is ``borrowed'' from the storage to perform a charge and the energy deficit in the storage can then be compensated from a subsequent discharge. This is possible provided that the amount of transferred energy and the involved time window are small. Note that the charging and discharging events intertwine on EVs at relatively high frequency. As long as the energy storage is sufficient, the reordering of charging and discharging would not disturb the energy ``flow'' along the energy path. The storage size would affect the system performance to a certain extent and we will leave the study of the impact of storage size for future investigation.

Note that an energy path should be loop-free. Although the composite vehicular routes are already loop-free, it is still possible to form an energy path with loops. However, the purpose of an energy path is to transfer energy. When the (dis)charging facilities at the road junctions are equipped with energy storage, loops in energy paths are meaningless and can only complicate  the organization and management of the system. Fig. \ref{fig:loop_storage} explains this with an energy path $p(s,t)=\langle r_1, r_2, r_3\rangle$. EVs on $r_1$ bring some energy from $s$ to $n$ and then the energy goes from $n$ along with the EVs on $r_2$ back to $n$. After that, EVs on $r_3$ bring the energy from $n$ to $t$. However, $r_2$ is redundant and it does not help transmit energy toward the destination. When being discharged at $n$ at the end of $r_1$, the energy can be charged on EVs along $r_3$ instead of $r_2$. Hence, this energy path can be re-constructed as $p(s,t)=\langle r_1, r_3\rangle$.

\begin{figure}[!t]
\centering
\includegraphics[width=3.2in]{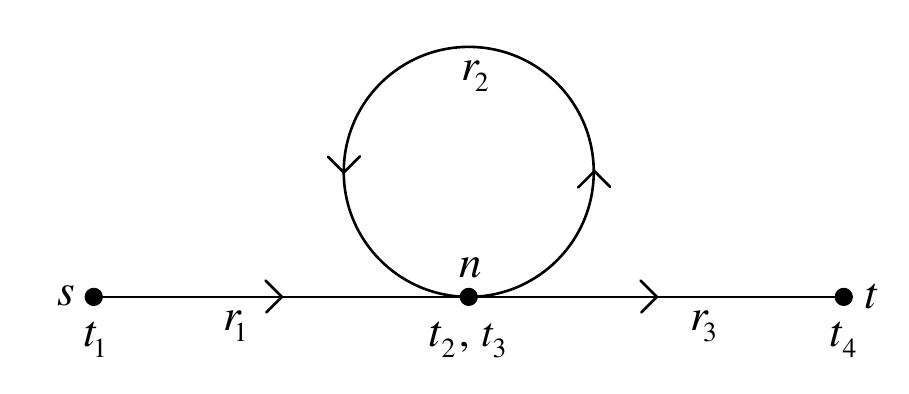}
\caption{Energy management for an energy path pretended to have a loop.}
\label{fig:loop_storage}
\end{figure}

\subsubsection{Delay}
When energy is transmitted along $p_j(s,t)$, it incurs some delay. Similar to the data network, there are ``propagation delay'', ``processing delay'', and ``transmission delay''. Propagation delay is the amount of time for the first dash of energy to travel from $s$ to $t$. Consider that an EV takes a delay of $d(a_k)$ to traverse the road connection $a_k\in \mathcal{A}$ on the average. For an EV to traverse Sub-route $r_i^j(n_i,m_i)$, it experiences a delay of $d(r_i^j(n_i,m_i))=\sum_{a_k\in r_i^j(n_i,m_i)}{d(a_k)}$ on the road connections. As an energy path is composed of a number of vehicular sub-routes, each of which is also composed of a number of road connections, the propagation delay of $p_j(s,t)$, denoted by $d(p_j)$, can be computed as 
\begin{align}
	d(p_j) = \sum_{i=1}^{|p_j|}{d(r_i^j(n_i,m_i))} = \sum_{i=1}^{|p_j|} \sum_{a_k\in r_i^j(n_i,m_i)}{d(a_k)}.
\end{align}

Processing delay refers to the time a (dis)charging facility takes to transfer the energy from one EV to another. Recall that dynamic (dis)charging happens instantaneously. Hence, we can assume negligible processing delay.

Transmission delay is the amount of time required to push all the energy from $s$ onto an energy path. It is related to the energy transfer rate of the energy path and we will discuss it next.

\subsubsection{Energy Transfer}
Let $f_i^j$ be the EV flow rate of the $i$-th segment of $p_j(s,t)$ and $w$ be the amount of energy carried by an EV in each charging-discharging cycle. EVs carry ``packets'' of energy and thus we call $w$ the ``packet size''. Then $w f_i^j$ is the energy transfer rate of $r_{i}^j(n_i,m_i)$. The overall energy transmission rate of $p_j(s,t)$, denoted by $g_j$, should be smaller than or equal to the minimum of the energy transmission rates of all its composite segments. Thus we have
\begin{align}
	g_j\leq wf_i^j, \quad i=1,\ldots,|p_j|. \label{transferRate}
\end{align}
\begin{figure}[!t]
\includegraphics[width=3.5in]{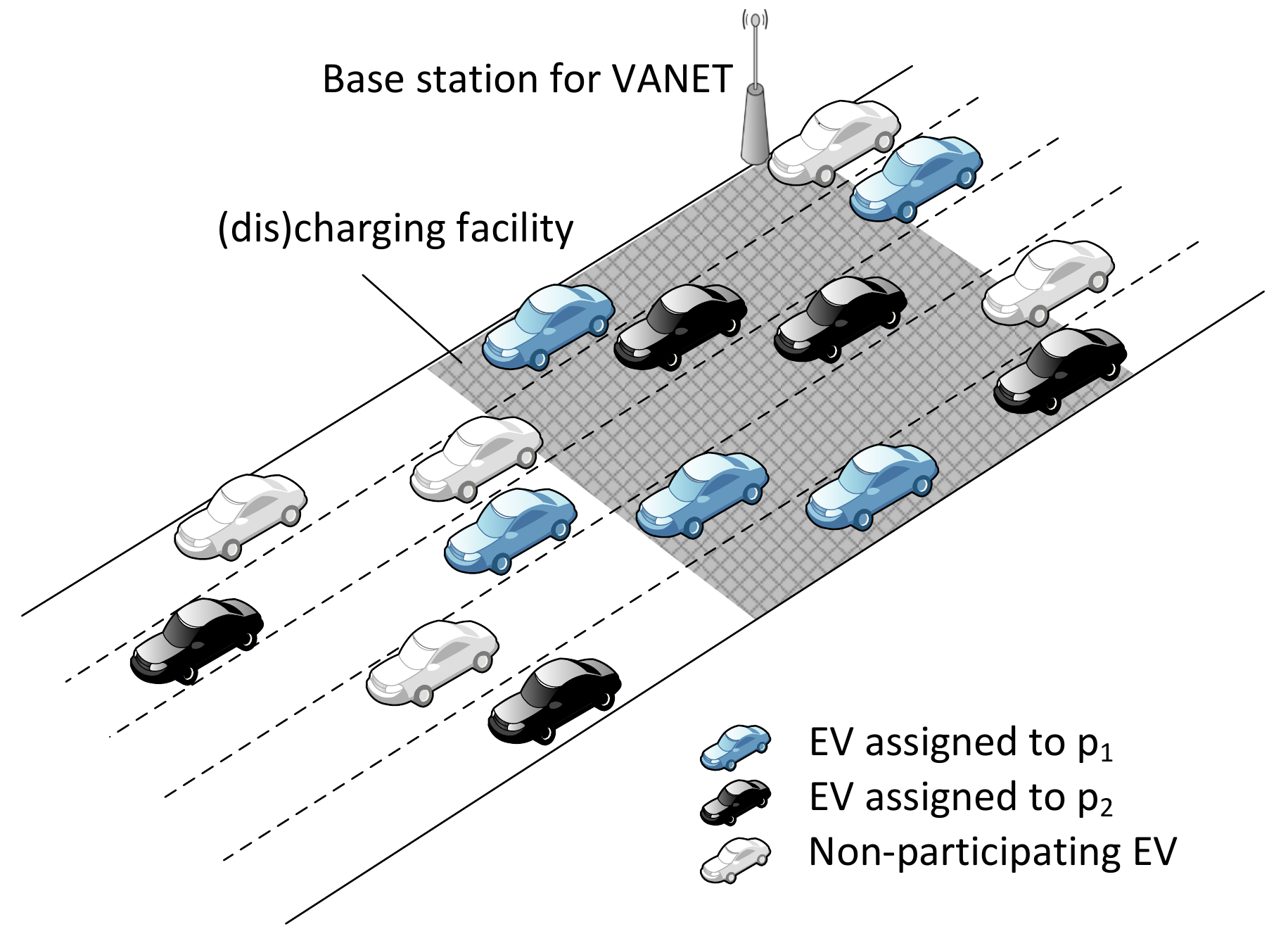}
\caption{Transmission rate assignment.}
\label{fig:rateAssign}
\end{figure}
Fig. \ref{fig:rateAssign} illustrates how we assign transmission rates to different energy paths in a real-world setting. It shows a road segment with a (dis)charging facility installed and it constitutes two energy paths $p_1$ and $p_2$. Suppose that this road segment belongs to the $i_1$-th and $i_2$-th sub-routes of $p_1$ and $p_2$, respectively. With VANET, we can differentiate the participating EVs from the non-participating ones and thus we can determine the vehicular flows of the participating EVs on this road segment. So $f_{i_1}^1=f_{i_2}^2$ is known. By suitably assigning the participating EVs to $p_1$ and $p_2$, we can configure $g_1$ and $g_2$.

When a charging or discharging event takes place, a certain fraction of energy will be lost. Let $z_c$ and $z_d$ be the charging and discharging efficiencies, respectively, where $0\leq z_c, z_d \leq 1$. During charging, only a fraction of $z_c$ can be successfully transferred to an EV from a charging facility and a fraction of $(1-z_c)$ is lost. A similar situation happens in discharging. Along $p_j(s,t)$, energy is charged $|p_j|$ times  and discharged $|p_j|$ times. Let $z=z_cz_d$. If we require $x_j$ units of energy to reach $t$ along $p_j(s,t)$, $\frac{x_j}{z^{|p_j|}}$ units of energy need to be injected from $s$, which has $(\frac{1}{z^{|p_j|}}-1)x_j$ units of energy loss. Hence the incurred transmission delay is $\frac{x_j}{z^{|p_j|}g_j}$ and the total time required is $d(p_j) + \frac{x_j}{z^{|p_j|}g_j}$. Let $T$ be the time window allowed for the energy transfer. The transferable amount of energy along $p_j(s,t)$ is governed by
\begin{align} 
x_j\leq (T-d(p_j)) z^{|p_j|}g_j. \label{transferamount}
\end{align}
 By considering all possible energy paths in $\mathcal{P}(s,t)$, the total amount of energy transferable from the source $s$ to the destination $t$, denoted by $x(s,t)$, in a time period $T$ is given by
\begin{align}
x(s,t) = \sum_{j|p_j\in \mathcal{P}(s,t)}{x_j} = \sum_{j|p_j\in \mathcal{P}(s,t)}{(T-d(p_j)) z^{|p_j|}g_j} \label{totalenergy}
\end{align}
and the corresponding energy loss is
\begin{align}
L(s,t) = \sum_{j|p_j\in \mathcal{P}(s,t)}{(\frac{1}{z^{|p_j|}}-1)x_j}. \label{energyLoss}
\end{align}


\section{Energy Path Construction} \label{sec:wholeset}

In order to route energy from a source to a destination, we need to establish an energy path connecting them  in VEN. Each energy path may be composed of a different number of vehicular segments experiencing different number of charging-discharging cycles.  So the amount of energy loss induced from each energy path varies. Moreover, each path may be constructed from a diverse subset of road connections with varying delays and thus it may admit different propagation delays. Therefore the choice of energy path for energy transfer affects the system performance.


\begin{figure}[!t]
	\begin{center}
		\subfigure[Example 1.]{\label{fig:example1}\includegraphics[width=3.2in]{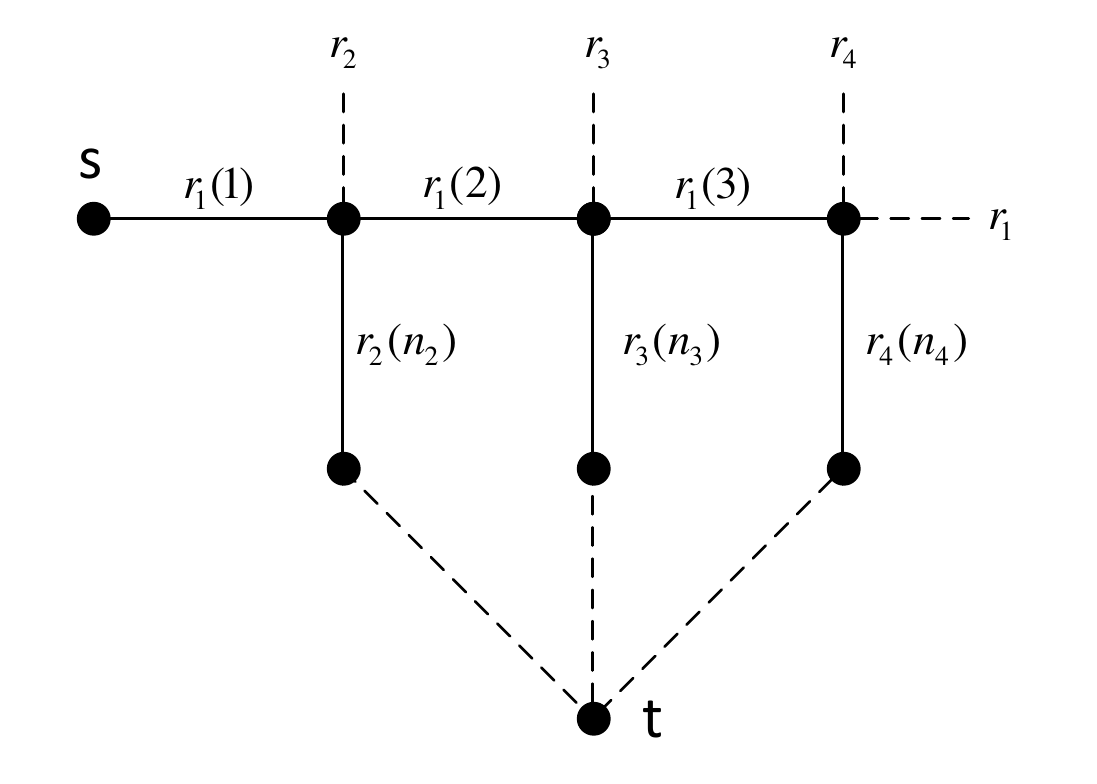}}
    \subfigure[Example 2.]{\label{fig:example2}\includegraphics[width=3.50in]{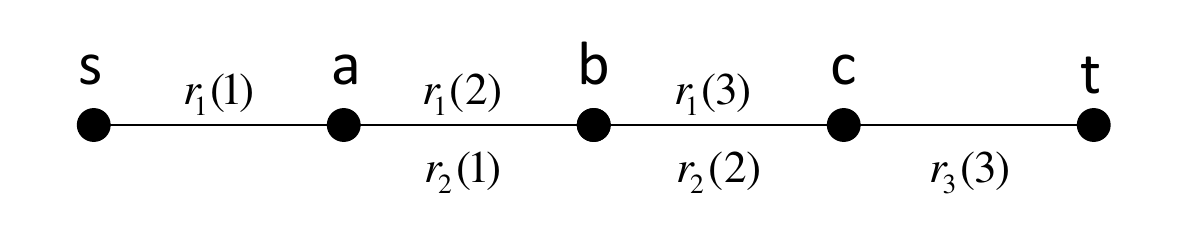}} 
	\end{center}
	\caption{Energy path formation.}
  \label{fig:routeformation}
\end{figure}

For each $(s,t)$ pair, many possible energy paths can be constructed. Consider the examples given in Fig. \ref{fig:routeformation}. In Fig. \ref{fig:example1}, there are four vehicular routes, $r_1$ to $r_4$. The source $s$ is attached to $r_1$ and the destination $t$ can be accessed through $r_2$, $r_3$, or $r_4$. In this case, we can utilize different segments of $r_1$ to construct energy paths. In other words, we have $\langle r_1(1),r_2(n_2),\ldots \rangle$, $\langle r_1(1,2),r_3(n_3),\ldots \rangle$, and $\langle r_1(1,3),r_4(n_4),\ldots \rangle$. In some cases with overlapping edges, multiple paths may result. Consider the example in Fig. \ref{fig:example2} with two vehicular routes, $r_1=\langle r_1(1),r_1(2),r_1(3) \rangle$ and $r_2=\langle r_2(1),r_2(2),r_2(3) \rangle$. We can connect segments of $r_1$ and $r_2$ at different connecting points (a, b, or c) to construct energy paths. Three paths can be constructed, i.e., $\langle r_1(1),r_2(1,3)\rangle$, $\langle r_1(1,2),r_2(2,3)\rangle$, and $\langle r_1(1,3),r_2(3)\rangle$. These examples illustrate the variety of  energy path construction.

Different scenarios have different performance requirements; in a given period $T$, we need to transmit $x(s,t)$ units of energy from the source $s$ to the destination $t$ subject to maximum energy loss of $L(s,t)$ units. All the variables $T$, $s$, $t$, $x(s,t)$, and $L(s,t)$ vary in different scenarios. In some cases, we may need to employ multiple energy paths in order to satisfy the performance requirements. For more information, the interested reader may refer to \cite{VEN}. For each $(s,t)$ pair,  determining the whole set of energy paths $\mathcal{P}(s,t)$ will ease planning of the energy transmission scheme. 
Consider a vector $\nu = [\nu_1, \ldots, \nu_{|\mathcal{P}(s,t)|}]$, such that $\sum_{j=1}^{|\mathcal{P}(s,t)|}{\nu_j} =1$ and $0\leq \nu_j \leq 1,$ for $j=1,\ldots,|\mathcal{P}(s,t)|$. Given $x(s,t)$, if we have $\mathcal{P}(s,t)$, the problem will be reduced to determining $\nu$ such that $x(s,t) = \sum_{j=1}^{||\mathcal{P}(s,t)|}\nu_j x_j$ with other performance requirements satisfied.
In other words, we are dealing with how the energy from the source is distributed on the available energy paths. 
In the rest of the paper, we will focus on routing between one source and one destination. Our results can be extended to the routing for multiple sources and/or destinations  and this will be left for future work. 
In the following, we discuss how to determine $\mathcal{P}(s,t)$.

\subsection{Construction of $\mathcal{P}(s,t)$}  \label{sec:construction}

\begin{algorithm}
\caption{Energy path construction} \label{pathalgo}
\begin{algorithmic}[1]
\STATE Construct $\tilde{G}(\tilde{\mathcal{N}},\tilde{\mathcal{A}})$ and $\mathcal{I}_{\tilde{\mathcal{A}}}$ from $G(\mathcal{N},\mathcal{A})$ with $\mathcal{R}$.
\STATE Construct the destination-not-accessiable node set $\overline{\mathcal{N}}$ from $G(\mathcal{N},\mathcal{A})$.
\STATE Construct $\tilde{\mathcal{A}}'$ from $\tilde{\mathcal{A}}$ with $\overline{\mathcal{N}}$.
\STATE Construct $\mathcal{K}=\{\langle s \rangle\}$ and $\hat{\mathcal{K}}=\varnothing$.
\REPEAT
	\FORALL {$i\in \tilde{\mathcal{N}}\setminus \{t\}$}
		\STATE Construct $\tilde{\mathcal{A}}'_i$ from $\tilde{\mathcal{A}}'$.
		\FORALL {$k\in\mathcal{K}$ with the end of sequence equal to $i$}
			\STATE Construct $\mathcal{K}'_{ki}$ with $\tilde{\mathcal{A}}'_i$.
			\STATE Update $\mathcal{K} \gets \mathcal{K}\cup \mathcal{K}'_{ki} \setminus k$.
		\ENDFOR
	\ENDFOR
	\FORALL {$k\in\mathcal{K}$ with the end of sequence equal to $t$}
		\STATE Update $\hat{\mathcal{K}} \gets \hat{\mathcal{K}}\cup k$.
	\ENDFOR
	\STATE Update $\mathcal{K} \gets \mathcal{K}\setminus \hat{\mathcal{K}}$.
\UNTIL {$\mathcal{K}=\varnothing$}
\STATE Construct $\mathcal{P}(s,t)$ from $\hat{\mathcal{K}}$ and $\mathcal{I}_{\tilde{\mathcal{A}}}$.
\end{algorithmic}
\end{algorithm}

We develop Algorithm \ref{pathalgo} to construct $\mathcal{P}(s,t)$. We start to construct a directed graph $\tilde{G}(\tilde{\mathcal{N}},\tilde{\mathcal{A}})$ and a collection of index sets $\mathcal{I}_{\tilde{\mathcal{A}}}=\{\mathcal{I}_{a}|a\in \tilde{\mathcal{A}}\}$ from $G(\mathcal{N},\mathcal{A})$ with $\mathcal{R}$ (Step 1), where $\tilde{\mathcal{N}}=\mathcal{N}$. For any $i,j\in \tilde{\mathcal{N}}$, $(i,j)$ is in $\tilde{\mathcal{A}}$ if there exists a vehicular route $r\in\mathcal{R}$ connecting $i$ and $j$. In other words, $i$ is connected to $j$ in $\tilde{G}(\tilde{\mathcal{N}},\tilde{\mathcal{A}})$ when there are EVs going from $i$ to $j$ along a particular vehicular route. $(i,j)$ represents accessibility of road junctions along at least one  vehicular route. 
Each $\mathcal{I}_a$ is an index set containing the indices of $r_l$'s, each of which contains $a$, i.e., $l\in \mathcal{I}_a \implies a\in r_l$. 
Then we determine the set of nodes $\overline{\mathcal{N}}$, each  of which cannot reach the destination $t$ on $G(\mathcal{N},\mathcal{A})$ (Step 2). This can be done by applying any shortest path algorithm, e.g. Dijkstra's algorithm \cite{dijkstra},  to each node in $G(\mathcal{N},\mathcal{A})$. A node $i$ is deemed not accessible to $t$ if there does not exist a path connecting $i$ to $t$ on  $G(\mathcal{N},\mathcal{A})$. We can ignore those nodes in $\overline{\mathcal{N}}$ in the subsequent manipulations because no paths connecting $s$ and $t$ can involve nodes in $\overline{\mathcal{N}}$.
Next we create an edge set $\tilde{\mathcal{A}}'$ by eliminating those edges from $\tilde{\mathcal{A}}$ with the starting nodes in $\overline{\mathcal{N}}$, i.e., $\tilde{\mathcal{A}}'=\{(i,j)\in \tilde{\mathcal{A}}|j\not\in \overline{\mathcal{N}}\}$. 

After that, we manipulate two sets, $\mathcal{K}$ and $\hat{\mathcal{K}}$ repeatedly (Steps 4--17), both of which contain sequences of nodes on $\tilde{G}(\tilde{\mathcal{N}},\tilde{\mathcal{A}})$. $\mathcal{K}$ maintains some developing sequences (partial energy paths) and those completed energy paths will be moved to  $\hat{\mathcal{K}}$ from $\mathcal{K}$. We initialize $\mathcal{K}$ with a single-node sequence $\langle s \rangle$ and $\hat{\mathcal{K}}$ as a null set (Step 4). In each iteration, we manipulate each node $i$ in $\tilde{\mathcal{N}}$ except the destination $t$.\footnote{Since loops are not allowed in any energy path, there does not exist an energy path with a composite sub-route originated from $t$.} We construct $\tilde{\mathcal{A}}'_i =\{(l,j)\in \tilde{\mathcal{A}}'|l=i\}$ by selecting those arcs with starting nodes equal to $i$ from $\tilde{\mathcal{A}}'$ (Step 7). We then check for each sequence $k$ in $\mathcal{K}$ ending with node $i$, i.e., $k=\langle k_1,\ldots, k_{|k|} | k_1=s, k_{|k|}=i \rangle$, where $|k|$ is the length of $k$ (Step 8). We create a set $\mathcal{K}'_{ki}$ of $|\tilde{\mathcal{A}}'_i|$ sequences by appending each $a=(i,j)\in \tilde{\mathcal{A}}'_i$ to $k$ as $\langle s,\ldots,i,j \rangle$ (Step 9). Then we update $\mathcal{K}$ by replacing $k$ with $\mathcal{K}'_{ki}$ (Step 10). After updating $\mathcal{K}$ with each $i$ in $\tilde{\mathcal{N}}$, we put those sequences in  $\mathcal{K}$ ending with the destination $t$ into $\hat{\mathcal{K}}$ (Steps 13--15). Then we remove those sequences already in $\hat{\mathcal{K}}$ from $\mathcal{K}$ (Step 16). The iterations terminate when there is no sequence in $\mathcal{K}$ (Step 17) and we output $\hat{\mathcal{K}}$.

In Step 18, we retrieve $\mathcal{P}(s,t)$ from $\hat{\mathcal{K}}$ and $\mathcal{I}_{\tilde{\mathcal{A}}}$. Each $k$ in $\hat{\mathcal{K}}$ represents at least one energy path (proof will be given in the next subsection) and potentially several. Note that each consecutive pair of nodes  $(i,j) $ in $k$, i.e., an arc in $\tilde{G}(\tilde{\mathcal{N}},\tilde{\mathcal{A}})$, can represent multiple vehicular sub-routes, which have been indexed with $\mathcal{I}_{\tilde{\mathcal{A}}}$. Those energy paths corresponding to $k=\langle k_1,\ldots, k_{|k|} \rangle$ can be determined by checking all combinations of the involved index sets $\mathcal{I}_{(k_1,k_2)}, \mathcal{I}_{(k_2,k_3)}, \ldots,\mathcal{I}_{(k_{|k|-1},k_{|k|})}$. For example, consider $k=\langle k_1,k_2,k_3\rangle$, $\mathcal{I}_{(k_1,k_2)}=\{i_1,i_2\}$, and $\mathcal{I}_{(k_2,k_3)}=\{i_3,i_4\}$. Then we have four index combinations, $[i_1,i_3]$, $[i_1,i_4]$,  $[i_2,i_3]$, and $[i_2,i_4]$. Each of these combinations gives one energy path, e.g., $[i_1,i_3]\implies \langle r_{i_1}(n_{i_1},m_{i_1}),r_{i_3}(n_{i_3},m_{i_3}) \rangle$ with $tail(r_{i_1}(n_{i_1}))=k_1$, $head(r_{i_1}(m_{i_1}))= tail(r_{i_3}(n_{i_3})) =k_2$, and $head(r_{i_3}(m_{i_3})) =k_3$. 
However, an index combination with repeated indices will not form an energy path; if an energy path is constituted from multiple sub-routes of the same vehicular route, we just need to route the energy from the very front of these sub-routes to the very end directly without incurring unnecessary energy loss from extra charging and discharging processes.
Each $k$ gives at most $|\mathcal{I}_{(k_1,k_2)}| \times |\mathcal{I}_{(k_2,k_3)}| \times \ldots \times |\mathcal{I}_{(k_{|k|-1},k_{|k|})}|$ energy paths. 
In this way, we can determine the whole set $\mathcal{P}(s,t)$ by examining all $k$ in $\hat{\mathcal{K}}$.

\subsection{Analytical Results} \label{sec:analysis}
Here we give the major analytical results related to Algorithm \ref{pathalgo}:
\begin{lemma} \label{lm:stop}
Algorithm \ref{pathalgo} must terminate.
\end{lemma}
\begin{proof}
When we manipulate $\mathcal{K}$ in the iterations, we replace each sequence $k$ in $\mathcal{K}$ with a set of sequences $\mathcal{K}_{ki}'$, each of which has the first $|k|$ elements exactly equal to $k$ and the $(|k|+1)$th element is based on $\tilde{\mathcal{A}}'_i$ such that the $|k|$-th and $(|k|+1)$-th elements constitute an arc found in $\tilde{\mathcal{A}}'_i$. $\tilde{\mathcal{A}}'_i$ is a subset of $\tilde{\mathcal{A}}$, which only contains those arcs whose ending nodes must be able to form a path to $t$ (See Step 3). In other words, the sequences in $\mathcal{K}$ keep elongating and eventually end with $t$. As we will move those sequences ending with $t$ to $\hat{\mathcal{K}}$, $\mathcal{K}$ must become empty eventually.
\end{proof}

\begin{lemma}
For any energy source $s$ and destination $t$, each sequence in $\hat{\mathcal{K}}$ constitutes at least an energy path $p_j(s,t)$.
\end{lemma}
\begin{proof}
In the iterations, we only extend the existing sequences in $\mathcal{K}$. Since $\langle s \rangle$ is the only sequence when $\mathcal{K}$ is initially defined, all sequences in $\mathcal{K}$ resulted in the subsequent iterations must start with $s$. When extending a sequence $\langle k_1,\ldots,k_i \rangle$ in $\mathcal{K}$ with a node $k_{i+1}$, $(k_i,k_{i+1})$ is an arc in $\tilde{G}(\tilde{\mathcal{N}},\tilde{\mathcal{A}})$, which means that Node $k_{i+1}$ can be accessed from Node $k_i$ on a particular vehicular route. At the end, only those sequences ending at $t$ can be found in $\hat{\mathcal{K}}$.

Each sequence $k=\langle k_1,k_2,\ldots,k_{|k|} \rangle$ with $k_1=s$ and $k_{|k|}=t$ in $\hat{\mathcal{K}}$ can then turn into an energy path $p_j(s,t)$. Each $(k_i,k_{i+1})$ pair, $i=1,\ldots,|k|-1$, represents the vehicular sub-route connecting junctions $k_i$ and $k_{i+1}$, i.e., $r_i(n_i,m_i)$ with $tail(r_i(n_i))=k_i$ and $head(r_i(m_i))=k_{i+1}$. Hence $k$ is constructed from $(|k|-1)$ sub-routes and $p_j(s,t)$ can be formed by concatenating the $(|k|-1)$ sub-routes.
\end{proof}

\begin{theorem} \label{thm:path}
For any energy source $s$ and destination $t$, Algorithm \ref{pathalgo} can determine the whole set of energy path $\mathcal{P}(s,t)$. 
\end{theorem}
\begin{proof}
Suppose that there exists an energy path generated from Sequence $k=\langle k_1,\ldots, k_l,k_{l+1},\ldots k_{|k|}\rangle$ which is not in $\hat{\mathcal{K}}$ when Algorithm \ref{pathalgo} terminates. That means, in a certain iteration, some segments of $k$, say $(k_l,k_{l+1})$ with $k_l=i$, are missing from $\tilde{\mathcal{A}}_i'$, and in turn missing from $\tilde{\mathcal{A}}'$, as all arcs with a starting node equal to $i$ have been used to construct $\mathcal{K}_{ki}'$. However, as long as $(k_l,k_{l+1})$ exists in $\tilde{\mathcal{A}}$ and a path can be formed from $k_{l+1}$ to $t$, $(k_l,k_{l+1})$ must be in $\tilde{\mathcal{A}}'$. This induces a contradiction and thus Algorithm \ref{pathalgo} can determine the whole $\mathcal{P}(s,t)$.
\end{proof}

\begin{corollary}
For any energy source $s$ and destination $t$, the number of energy paths in $\mathcal{P}(s,t)$ is finite.
\end{corollary}
\begin{proof}
From Theorem \ref{thm:path}, the whole $\mathcal{P}(s,t)$ can be determined when Algorithm \ref{pathalgo} terminates. From Lemma \ref{lm:stop}, the algorithm must terminate after a certain number of iterations, each of which can only manipulate a finite number of sequences in $\mathcal{K}$. Thus the number of sequences in $\hat{\mathcal{K}}$ is finite when the iterations terminate.
As $\mathcal{R}$ is finite, $\mathcal{I}_{\tilde{\mathcal{A}}}$ is also finite. Thus the number of energy paths constructable from each $k$ in $\hat{\mathcal{K}}$ must be finite.
Hence the number of energy paths in $\mathcal{P}(s,t)$ is finite.
\end{proof}

\begin{lemma} \label{lm:length}
An energy path $p_j(s,t)$ on $G(\mathcal{N},\mathcal{A})$ is composed of a maximum of $(|\mathcal{N}|-1)$ arcs, constituted from at most $(|\mathcal{N}|-1)$ vehicular routes.
\end{lemma}
\begin{proof}
Since loops are not allowed, the longest possible path from $s$ to $t$ is a Hamiltonian path in $G(\mathcal{N},\mathcal{A})$, which has at most $(|\mathcal{N}|-1)$ arcs.  For this longest path, the worst case is that each arc is originated from a different vehicular route. Hence, an energy path can be constituted from at most $(|\mathcal{N}|-1)$ vehicular routes.
\end{proof}

Define 
\begin{align} \label{fn}
f(n)=\begin{cases}
1+(n-1)f(n-1) & \mbox{if } n>1,\\
1							& \mbox{if } n=1,\\
0							& \mbox{otherwise}. 
\end{cases}
\end{align}
\begin{lemma} \label{lm:cardP}
The cardinality of $\hat{\mathcal{K}}$ is upper bounded by $f(|\mathcal{N}|-1)$.
\end{lemma}
\begin{proof}
Each energy path can be represented by a sequence of nodes based on $\tilde{G}(\tilde{\mathcal{N}},\tilde{\mathcal{A}})$. Since an energy path does not contain loops, no identical nodes appear in the sequences in $\hat{\mathcal{K}}$. The largest quantity of energy paths happens when $\tilde{G}(\tilde{\mathcal{N}},\tilde{\mathcal{A}})$ is a complete graph. We now consider the complete graph in the following. 

We start with the node $s$. Since $s$ is connected to every node in $\mathcal{N}\setminus\{s\}$, We can extend $\langle s\rangle$ to form $(|\mathcal{N}|-1)$ new sequences, given by $\langle s, i\rangle,\forall i\in \mathcal{N}\setminus\{s\}$. One of these sequence is $\langle s, t\rangle$ and we keep on extend the rest of the $(|\mathcal{N}|-2)$ sequences. For each of these sequences, we can extend it with an additional node with $(|\mathcal{N}|-2)$ possible choices. After extension, one of them ends with the node $t$ and we continue to extend the rest of the resultant sequences. Hence the cardinality of $\hat{\mathcal{K}}$ is represented by $f(|\mathcal{N}|-1)$.
\end{proof}

\begin{theorem} \label{thm:runtime}
The total running time of Algorithm \ref{pathalgo} is bounded by $\mathcal{O}(|\mathcal{N}|^{|\mathcal{N}|-1}(|\mathcal{N}|-1)^{|\mathcal{N}|}(|\mathcal{N}|-2)!)$.
\end{theorem}

\begin{proof}
The algorithm can be divided into three parts, namely Steps 1--4, Steps 5--17, and Step 18.

The first part is computed once. In Step 1, we add arcs to $G(\mathcal{N},\mathcal{A})$ to form $\tilde{G}(\tilde{\mathcal{N}},\tilde{\mathcal{A}})$ based on $\mathcal{R}$. For each vehicular route in $\mathcal{R}$, there are at most $|\mathcal{A}|$ arcs and we can create at most $(|\mathcal{N}|-1)$ additional arcs. Each node in $G(\mathcal{N},\mathcal{A})$ can be the origin of at most $(|\mathcal{N}|-1)$ vehicular routes, each of which ends at a distinct node. So we have $|\mathcal{R}|\leq |\mathcal{N}|(|\mathcal{N}|-1)$. Therefore the running time of Step 1 is $\mathcal{O}(|\mathcal{N}|^3)$.

Step 2 is equivalent to solving a single-destination shortest-path problem. By Dijkstra's algorithm, its running time is $\mathcal{O}(|\mathcal{N}|^2)$ \cite{dijkstra}.

In Step 3, we eliminate arcs from $\tilde{\mathcal{A}}$. As $|\tilde{\mathcal{A}}|\leq |\mathcal{N}|(|\mathcal{N}|-1)$, the running time of Step 3 is $\mathcal{O}(|\mathcal{N}|^2)$.

Trivially, Step 4 takes a running time of $\mathcal{O}(1)$.

The second part contains a repeat-loop (Steps 5--17). Each loop is further composed of two sub-parts, i.e., Steps 6--12 and Steps 13--16. The former examines each processing sequence $k$ in $\mathcal{K}$ and replaces it with a set of new sequences, each of which is formed by extending $k$ along a possible arc from $\tilde{A}$. This is equivalent to enumerating all possible sequences. By Lemma \ref{lm:cardP}, there are at most $f(|\mathcal{N}|-1)$ sequences. By Lemma \ref{lm:length}, each sequence involves at most $(|\mathcal{N}|-1)$ extensions when appending each node one by one to the sequence.
Eq. \eqref{fn} can be written as
\begin{align*}
f(n) =& 1+(n-1)+(n-1)(n-2)+\ldots+\\
			&(n-1)\times\ldots\times 2 + (n-1)\times\ldots\times 2 \\
		=& \frac{(n-1)!}{(n-1)!} + \frac{(n-1)!}{(n-2)!}+ \ldots + \frac{(n-1)!}{1!}+\frac{(n-1)!}{1!}\\
		=&(n-1)! \sum_{i=0}^{n-1}{\frac{1}{i!}}\\
		\leq& (n-1)!e.
\end{align*}
So the number of operations involved in this sub-part is upper bounded by $(|\mathcal{N}|-2)!e(|\mathcal{N}|-1)$.
 This sub-part has a running time of $\mathcal{O}((|\mathcal{N}|-2)!|\mathcal{N}|)$.

The second sub-part is to move any complete sequences with the end nodes equal to $t$ to $\hat{\mathcal{K}}$ from $\mathcal{K}$. As there are at most $f(|\mathcal{N}|-1)$ sequences, it has a running time of $\mathcal{O}((|\mathcal{N}|-2)!)$.

The third part (Step 18) is to construct energy paths from each $k$ in $\hat{\mathcal{K}}$. To do this, for each $k=\langle k_1,\ldots,k_{|k|}\rangle$, we create all combinations of the index sets of the involved sub-routes from $\mathcal{I}_{\tilde{\mathcal{A}}}$, i.e., $\{\mathcal{I}_{(k_1,k_2)},\mathcal{I}_{(k_2,k_3)},\ldots,\mathcal{I}_{(k_{|k|-1},k_{|k|})} \}$. The number of combinations created for each $k$ is 
\begin{align*}
|\mathcal{I}_{(k_1,k_2)}|\times|\mathcal{I}_{(k_2,k_3)}|\times \ldots \times |\mathcal{I}_{(k_{|k|-1},k_{|k|})}| \leq |\mathcal{R}|^{|\mathcal{N}|-1}\\
\leq |\mathcal{N}|^{|\mathcal{N}|-1}(|\mathcal{N}|-1)^{|\mathcal{N}|-1}.
\end{align*}
As each combination contains at most $|\mathcal{N}|-1$ indices and there are at most $f(|\mathcal{N}|-1)$ sequences in $\hat{\mathcal{K}}$, the running time of this part is $\mathcal{O}(|\mathcal{N}|^{|\mathcal{N}|-1}(|\mathcal{N}|-1)^{|\mathcal{N}|}(|\mathcal{N}|-2)!)$.

Therefore, the total running time of Algorithm \ref{pathalgo} is bounded by $\mathcal{O}(|\mathcal{N}|^{|\mathcal{N}|-1}(|\mathcal{N}|-1)^{|\mathcal{N}|}(|\mathcal{N}|-2)!)$.
\end{proof}

\subsection{Utilization of $\mathcal{P}(s,t)$} \label{subsec:optimization}

The purpose of Algorithm \ref {pathalgo} is to enumerate all possible energy paths connecting $s$ and $t$. However,
to transmit energy from $s$ to $t$, after constructing the energy paths, we also need configure the energy transmission scheme by determining how much and how fast energy should be transported along the energy paths. In other words, we assign the energy transmission rate $g_j$ for each energy path $j$ in the allowed time $T$. 
As discussed in \cite{VEN}, different objectives are possible when utilizing VEN for energy transmission. For example, we may decide to maximize the total amount of energy transferred in a given time window or minimize the total energy loss with certain quantity of transferred energy guarantee. The energy transmission rate assignment relies on the system objective. 

As explained in \cite{VEN}, we can determine the complete energy transmission scheme to achieve various objectives systemically in the form of optimization. Determination of  $\mathcal{P}(s,t)$ can facilitate the formulations. 
With the whole set of energy paths $\mathcal{P}(s,t)$, the problem of configuring the network to fulfill a certain transmission objective can be reduced to assigning the energy transmission rates for all possible energy paths. If an energy path is not required, we can just assign its transmission rate with a zero value. 
Let $h_a$ be the vehicular flow of road connection $a\in \mathcal{A}$ and $\underline{X}$ be the energy target.  We demonstrate the utilization of $\mathcal{P}(s,t)$ with the example of minimizing the total energy loss with a  guaranteed transferable amount of energy, as follows: 
\begin{subequations}
\label{minopt}
\begin{align}
\text{minimize}\quad 	& \sum_{j=1}^{|\mathcal{P}(s,t)|}{(\frac{1}{z^{|p_j|}}-1)   x_j} \label{opt:obj}\\
\text{subject to}\quad 
& 0\leq x_j\leq (T-d(p_j))z^{|p_j|}g_j, \quad j=1,\ldots, |\mathcal{P}(s,t)| \label{opt:con1}\\
& 0\leq g_j \leq w f_i^j, \quad i=1,\ldots,|p_j|, j=1,\ldots, |\mathcal{P}(s,t)| \label{opt:con2}\\
& \sum_{j|a\in p_j} \frac{g_j}{w} \leq h_a, \quad a\in \mathcal{A} \label{opt:con3}\\
& \sum_{j=1}^{|\mathcal{P}(s,t)|}{x_j} \geq \underline{X} \label{opt:con4}
\end{align}
\end{subequations}
We minimize the total incurred energy loss in \eqref{opt:obj} based on \eqref{energyLoss}. For each energy path $j$ in $\mathcal{P}(s,t)$, \eqref{opt:con1}, from \eqref{totalenergy}, limits the amount of energy transmitted along  $j$ with rate $g_j$ in a time period of duration $T$. \eqref{opt:con2} defines $g_j$ based on the packet size $w$ and the vehicular flows $f_i^j$. When multiple energy paths share a road connection, \eqref{opt:con3} ensures that each connection has sufficient car flow to support all the involved energy paths. \eqref{opt:con4} ensures that the energy transmission target is satisfied. In \eqref{minopt}, $z$, $T$, $f_i^j$, $w$, $h_a$, and $\underline{X}$ are system parameters. When $\mathcal{P}(s,t)$ is given, we can determine $|p_j|$ and $d(p_j)$. $x_j$ and $g_j$ are the only variables of the problems. It can be seen that \eqref{minopt} is a linear program (LP) and it can be easily solved by a standard LP solver.

As the number of possible energy paths is generally huge, it is not surprising that the complexity of Algorithm \ref{pathalgo} is even greater than factorial time. For a fairly large network, it may be difficult to enumerate the whole $\mathcal{P}(s,t)$. In fact, $\mathcal{P}(s,t)$ confines the scope of search in the feasible region of \eqref{minopt}. If we only have a subset of energy paths, denoted by $\mathcal{P}'\subset\mathcal{P}(s,t)$, we can construct a similar problem as \eqref{minopt} with $\mathcal{P}'$. This problem is still an LP and easy to be solved. The optimal solution deduced from $\mathcal{P}'$ is in fact a sub-optimal solution of the original problem given in \eqref{minopt} with $\mathcal{P}(s,t)$. 
It is easy to modify Algorithm \ref{pathalgo} to construct a subset of $\mathcal{P}(s,t)$.
Therefore, when time is insufficient to construct the whole $\mathcal{P}(s,t)$, and we can only get a subset of $\mathcal{P}(s,t)$, the methodologies of solving most VEN problems can still carry through and we can still obtain sub-optimal solutions.

\subsection{Discussion}
When configuring the whole $\mathcal{P}(s,t)$, we ignore the details of  vehicular flows of the underlying road connections. This works fine as we can assign $g_j$ for all $j\in \mathcal{P}(s,t)$ systemically in terms of their optimality after determining $\mathcal{P}(s,t)$. The reason why we can do so is that all ``interactions'' among energy paths can be taken into account when conducting transmission rate assignment with all energy paths known.
%
When $\mathcal{P}(s,t)$ is pre-determined, the optimization approach discussed in Section \ref{subsec:optimization} allows us to decide the \textit{optimal} energy transmission rate for each energy path easily. 
However,  from Lemma \ref{lm:cardP}, the size of $\mathcal{P}(s,t)$ can grow super-exponentially with the size of the network and thus  determining $\mathcal{P}(s,t)$ is not trivial, especially when the network is large. In such cases, we may only be able to construct some of the energy paths instead of the whole $\mathcal{P}(s,t)$. 
This allows us to obtain sub-optimal transmission rates by solving the optimization problem with a subset of $\mathcal{P}(s,t)$. However, the performance of this method depends on the ``quality'' of the chosen subset. As some ``good'' energy paths may not have been included in the chosen subset of  energy paths, we cannot determine the best configurations of the energy paths to transmit energy. In this approach, when deciding the subset, we do not take their characteristics into account. 
To strive for better performance, we may need to construct the required energy paths with the consideration of the underlying vehicular flows and other information.

\section{Heuristic for the Power Loss Minimization Problem} \label{sec:heuristic}

To solve the power loss minimization problem given in \eqref{minopt}, we need to determine a set of energy paths $\mathcal{P}'=\{p_j\}$ and their corresponding energy transmission rates $g_j$. Then, by \eqref{totalenergy} and \eqref{energyLoss}, we determine the  transferred energy and energy loss accordingly. $\mathcal{P}'$ need not be the whole $\mathcal{P}(s,t)$ as long as the energy paths in $\mathcal{P}'$ have the properties required for solving \eqref{minopt}.
In this section,  we develop a heuristic to solve \eqref{minopt} by constructing $\mathcal{P}'$ and assigning their corresponding $g_j$ at the same time, based on the properties of the energy paths.

\begin{figure}[!t]
\centering
\includegraphics[width=1.0in]{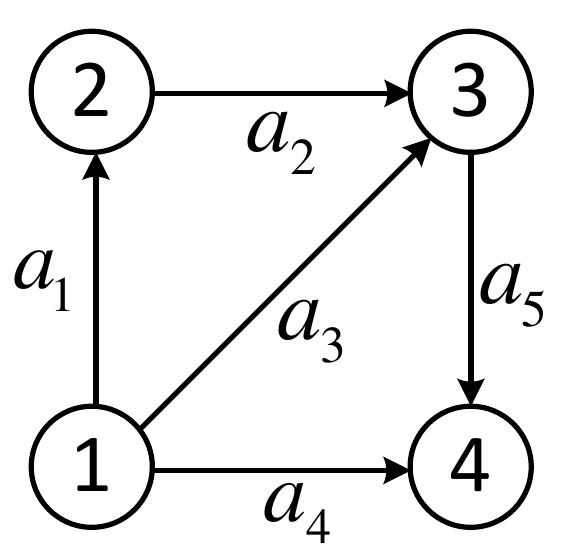}
\caption{A vehicular network of four nodes and five arcs.}
\label{fig:example}
\end{figure}

From \eqref{transferRate}, we can see that $g_j$ depends on the vehicular flows $f_i^j$ of all the composite road connections $a_i$ along Route $j$. However, $a_i$ may be used to construct multiple energy paths simultaneously. Depending on the packet size $w$, the flow of $a_i$ may need to be split and shared by the corresponding energy paths. In other words, the transmission rate assignment of one energy path may affect that of another path. 
Consider the illustrative example given in Fig. \ref{fig:example}, which shows a vehicular network of four nodes and five arcs. Suppose that we transmit energy from Nodes 1 to 4 along three energy paths, $p_1(1,4)=\langle a_1,a_2, a_5 \rangle$, $p_2(1,4)=\langle a_3, a_5 \rangle$, and $p_3(1,4)=\langle a_4 \rangle$. As both $p_1(1,4)$ and $p_2(1,4)$ contain $a_5$, the vehicular flow of $a_5$ may need to be shared by both paths. If $p_2(1,4)$ occupies too much flow of $a_5$, the residue flow of $a_5$ for $p_1(1,4)$ may become a bottleneck such that the vehicular flows available along $a_1$ and $a_2$ may not have been fully utilized. When assigning energy transmission rates $g_1$ and $g_2$, all vehicular flows of $a_1$, $a_2$, $a_3$, and $a_5$ need to be considered.
We may reduce the ``interactions'' of energy paths by avoiding using those paths with road connections being shared. Therefore,  transmission rate assignment may need to be considered together with energy path construction.

Here we focus on large networks where the whole set of energy paths are not economical to compute. We aim to construct only the necessary subset of energy paths and assign the corresponding energy transmission rates to fulfill the system design objective. 
In this section, We aim to minimize the energy loss when transmitting energy to the destination from the source. 

Basically, the heuristic determines a set of energy paths which experience the least charging-discharging cycles to transfer energy. As seen from \eqref{energyLoss}, energy loss of a path is proportional to the amount of energy transferred, i.e., $x_j$, and inversely proportional to the number of  charging-discharging cycles experienced, i.e., $|p_j|$. $\tilde{G}(\tilde{\mathcal{N}},\tilde{\mathcal{A}})$ allows us to find these paths; the number of hops possessed by a path on $\tilde{G}(\tilde{\mathcal{N}},\tilde{\mathcal{A}})$  represents the number of charging-discharging cycles experienced by the corresponding energy path. So we utilize the energy path with the least number of cycles to transfer as much energy as possible. Then we consider  the one with the next least number of cycles and so on until we have reached the energy target $\underline{X}$.

\begin{algorithm}
\caption{Heuristic for power loss minimization} \label{ratealgo}
\begin{algorithmic}[1]
\STATE Construct $\tilde{G}(\tilde{\mathcal{N}},\tilde{\mathcal{A}})$ from $G(\mathcal{N},\mathcal{A})$ with $\mathcal{R}$.
\STATE Define $\mathcal{P}_{L}:=\varnothing$ and $\psi:=0$
\STATE $done :=0$
\REPEAT
	\STATE Determine the energy path $p_j$ with the least number of hops from $s$ to $t$ on $\tilde{G}(\tilde{\mathcal{N}},\tilde{\mathcal{A}})$.
	\STATE Set $\delta := \inf\{f_i^j|r_i^j\in p_j\}$.
	\STATE Set $g_j :=w\delta$ and $x_j:=(T-d(p_j))z^{|p_j|}g_j$.
	\IF {$\psi + x_j< \underline{X}$}
		\STATE Update $\psi:=\psi + x_j$.
		\STATE Update $f_i^j:=f_i^j-\delta$ for all $r_i^j\in p_j$.
		\STATE Update $\mathcal{R}$ and $\tilde{G}(\tilde{\mathcal{N}},\tilde{\mathcal{A}})$.
	\ELSE
		\STATE Set $done:=1$
		\STATE Set $x_j:=\underline{X}-\psi$
		\STATE Set $g_j:=\frac{x_j}{(T-d(p_j))z^{|p_j|}}$
	\ENDIF
	\STATE Update $\mathcal{P}_{L}:=\mathcal{P}_{L}\cup p_j$
\UNTIL {$done=1$}
\RETURN $\mathcal{P}_{L}$ and $\{g_j\}$.
\end{algorithmic}
\end{algorithm}

Algorithm \ref{ratealgo} illustrates the implementation details of the heuristic. Similar to Algorithm \ref{pathalgo}, we first construct 
$\tilde{G}(\tilde{\mathcal{N}},\tilde{\mathcal{A}})$ from $G(\mathcal{N},\mathcal{A})$ with $\mathcal{R}$ (Step 1). Then we initialize $\mathcal{P}_{L}$ for storing the constructed energy paths and $\psi$ for counting the amount of energy which can reach the destination $t$ along the energy paths found in $\mathcal{P}_{L}$ (Step 2). We define a flag $done$ for the repeat-loop next (Step 3). We construct energy paths iteratively until the amount of energy transmittable to $t$, i.e. $\psi$, is greater than or equal to the requested amount $\underline{X}$ (Steps 4--18). In the $j$-th iteration, we first determine the shortest path, in terms of number of hops, from $s$ to $t$ on $\tilde{G}(\tilde{\mathcal{N}},\tilde{\mathcal{A}})$ (Step 5). As each edge of $\tilde{G}(\tilde{\mathcal{N}},\tilde{\mathcal{A}})$ represents a vehicular sub-route, we can construct the energy path $p_j$ by concatenating the corresponding vehicular sub-routes.  Then we set $\delta$ as the vehicular flow $f_i^j$ of the sub-route with the minimum flow along $p_j$ (Step 6). Based on \eqref{transferRate} and \eqref{transferamount}, we determine the maximum possible energy transmission rate $g_j = w\delta$ and the transferable amount of energy $x_j$ (Step 7). 
Next we check if the cumulative transferred energy is still smaller than the energy target $\underline{X}$ (Step 8). If so, we update the total transferable amount of energy by including the amount $x_j$ from $p_j$ (Step 9). After that, we update the vehicular flow $f_i^j$ of each sub-route $r_i^j$ along $p_j$ by subtracting the occupied flow $\delta$ (Step 10). We also update $\mathcal{R}$ and $\tilde{G}(\tilde{\mathcal{N}},\tilde{\mathcal{A}})$ (Step 11) as follows: For those $r_i=\langle a_1^i, \ldots, a_{n_i}^i,\ldots,a_{m_i}^i,\ldots a_{|r_i|^i}\rangle$ in which any sub-routes appeared in $p_j$ have zero flow (say $r_i(n_i,m_i)$), we truncate the segment starting from $r_i(n_i)$ and $r_i$ becomes $\langle a_1^i, \ldots, a_{n_i-1}^i\rangle$. We then re-construct $\tilde{G}(\tilde{\mathcal{N}},\tilde{\mathcal{A}})$ with the updated $\mathcal{R}$. 
If we have accumulated enough transferred energy, $p_j$ is the last energy path required and we set the flag $done=1$ (Step 13). For the last path, we do not need to transfer at its maximum capacity, as determined in Step 7. The amount of energy  to be transferred on $p_j$ is the residual amount, i.e., $\underline{X}-\psi$ (Step 14) and the required transmission rate is determined based on $x_j$ (Step 15) accordingly.
We include $p_j$ into $\mathcal{P}_L$ (Step 17). Finally, we output $\mathcal{P}_{L}$ and the corresponding transmission rate $g_j$ as the solution for the problem (Step 19).


%
\section{Performance Evaluation} \label{sec:performance}

By abuse of notation, we denote the set of energy paths connecting a particular source and destination without specifying $s$ and $t$ by $\mathcal{P}$. We have introduced three methods for VEN routing:
\begin{itemize}
	\item Method I: the optimization-based approach with the whole $\mathcal{P}$ explained in Section \ref{sec:wholeset};
	\item Method II: the optimization-based approach with a partial $\mathcal{P}$ discussed in Section \ref{sec:wholeset}; and
	\item Method III: the heuristic proposed in Section \ref{sec:heuristic}.
\end{itemize} 
We will evaluate their performance by applying them to \eqref{minopt}. Before that, we investigate the growth of cardinality of $\mathcal{P}$, which allows us to obtain more insight to differentiate Methods I and II. 

\subsection{Growth of $|\mathcal{P}|$}
In general, as explained in Section \ref{sec:analysis},  $|\mathcal{P}|$ grows with $|\mathcal{N}|$. Moreover, an increase of the level of vehicular information disclosure results in longer vehicular routes and thus more energy paths will be produced. We examine these factors for the growth of  $|\mathcal{P}|$.

\begin{figure*}[!t]
	\begin{center}
		\subfigure[4-node network]{\label{fig:rand4}\includegraphics[width=3.2in]{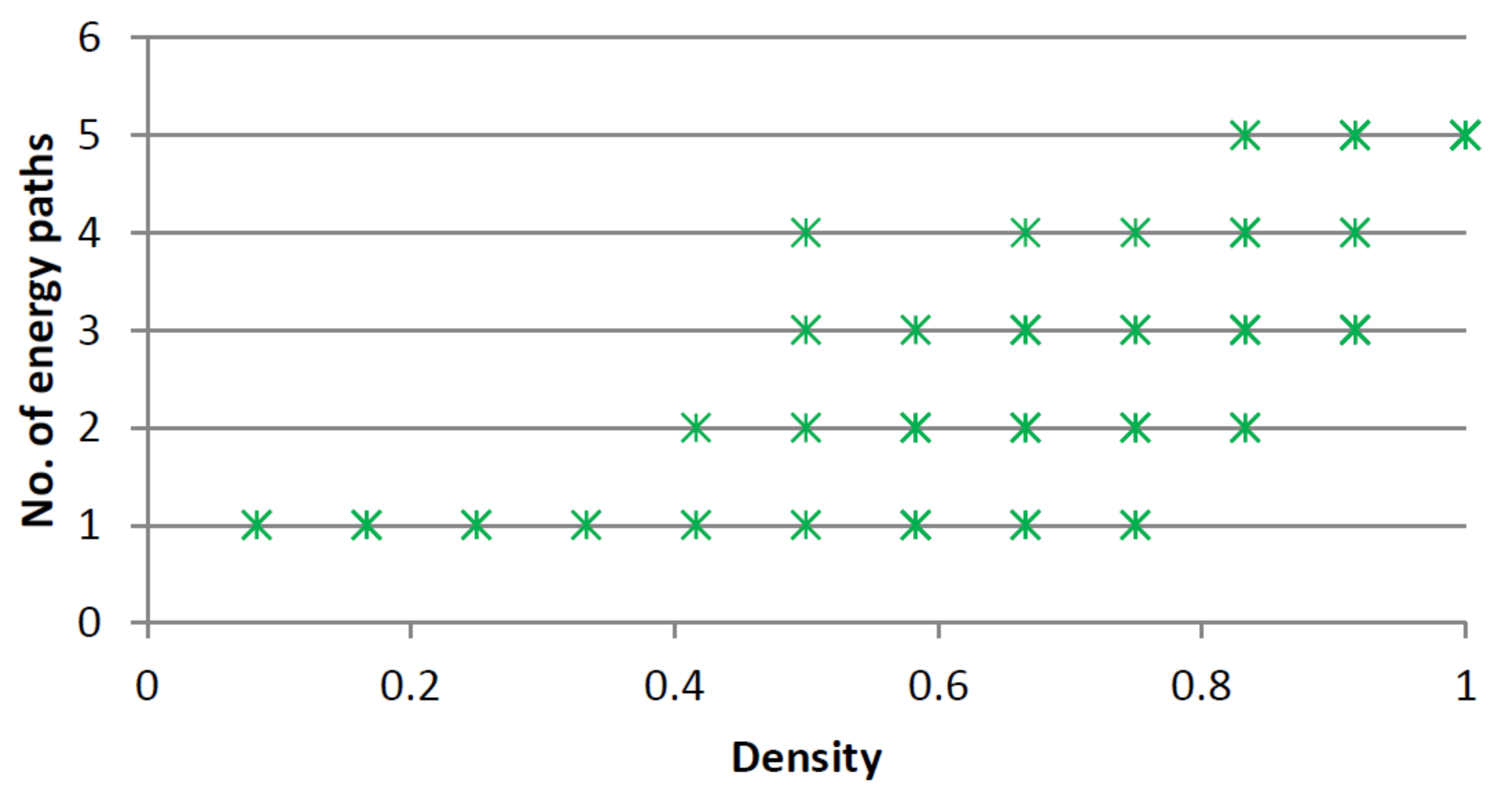}}
    \subfigure[6-node network]{\label{fig:rand6}\includegraphics[width=3.2in]{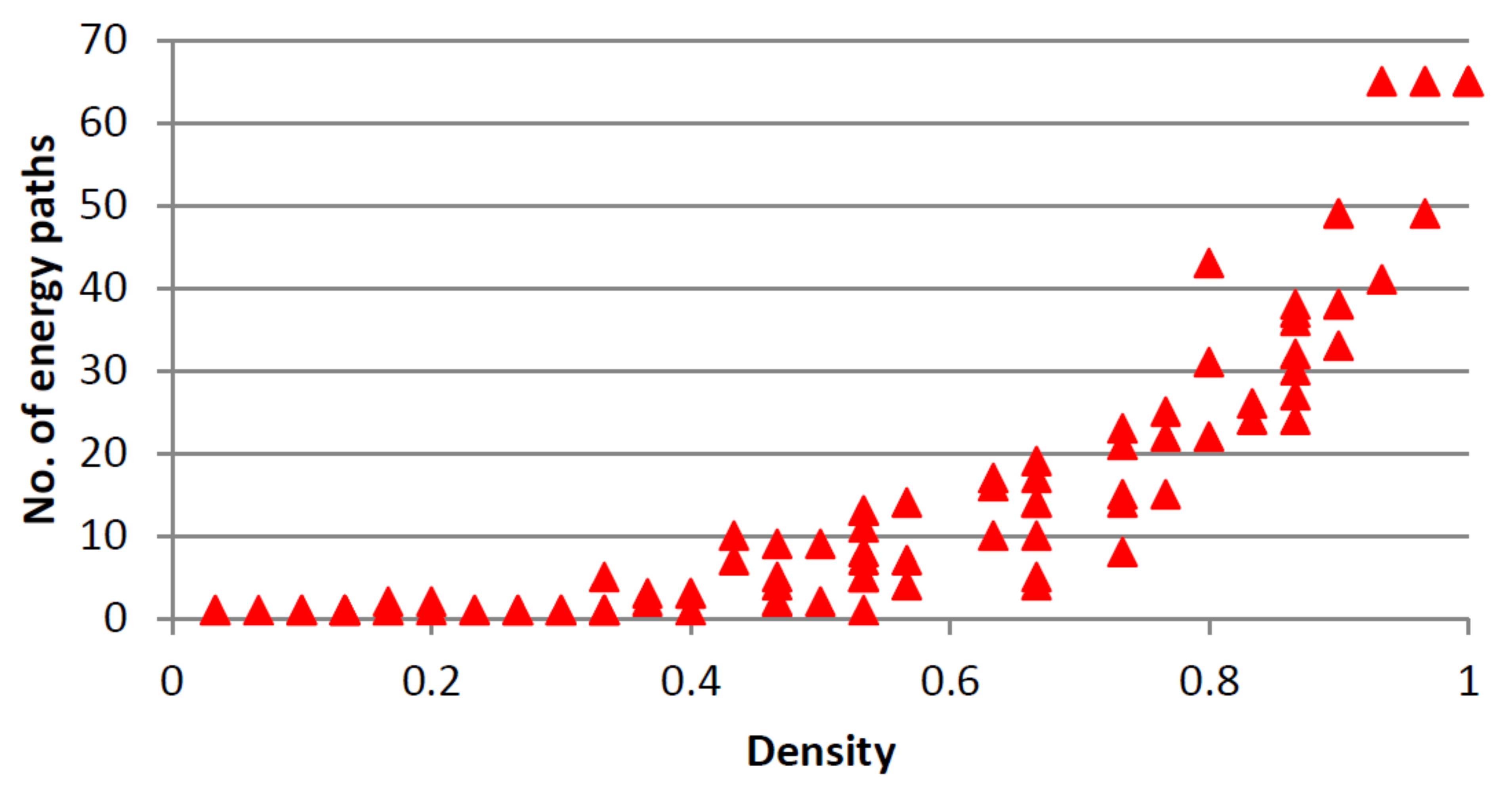}} 
		\subfigure[8-node network]{\label{fig:rand8}\includegraphics[width=3.2in]{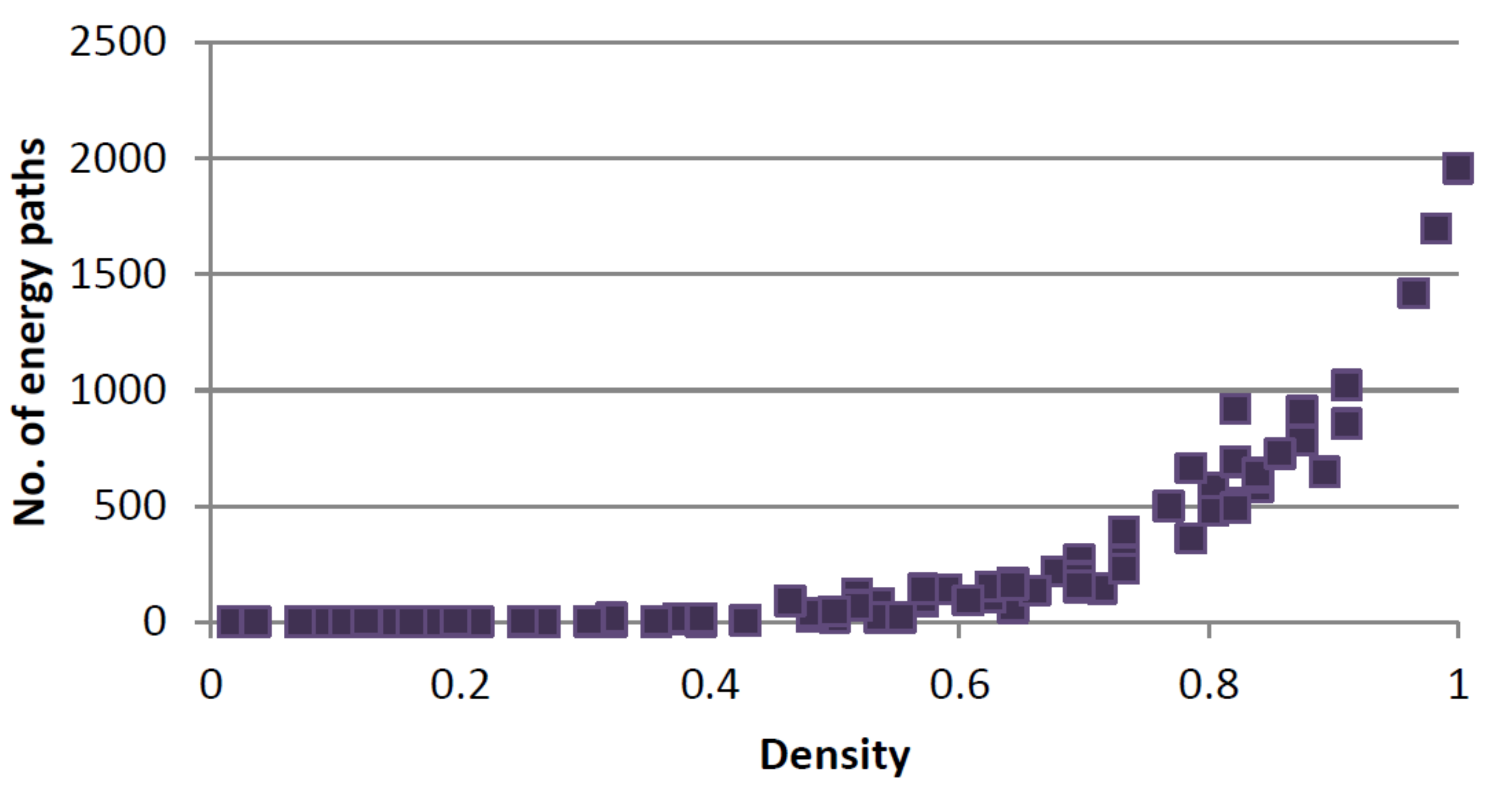}}
    \subfigure[10-node network]{\label{fig:rand10}\includegraphics[width=3.2in]{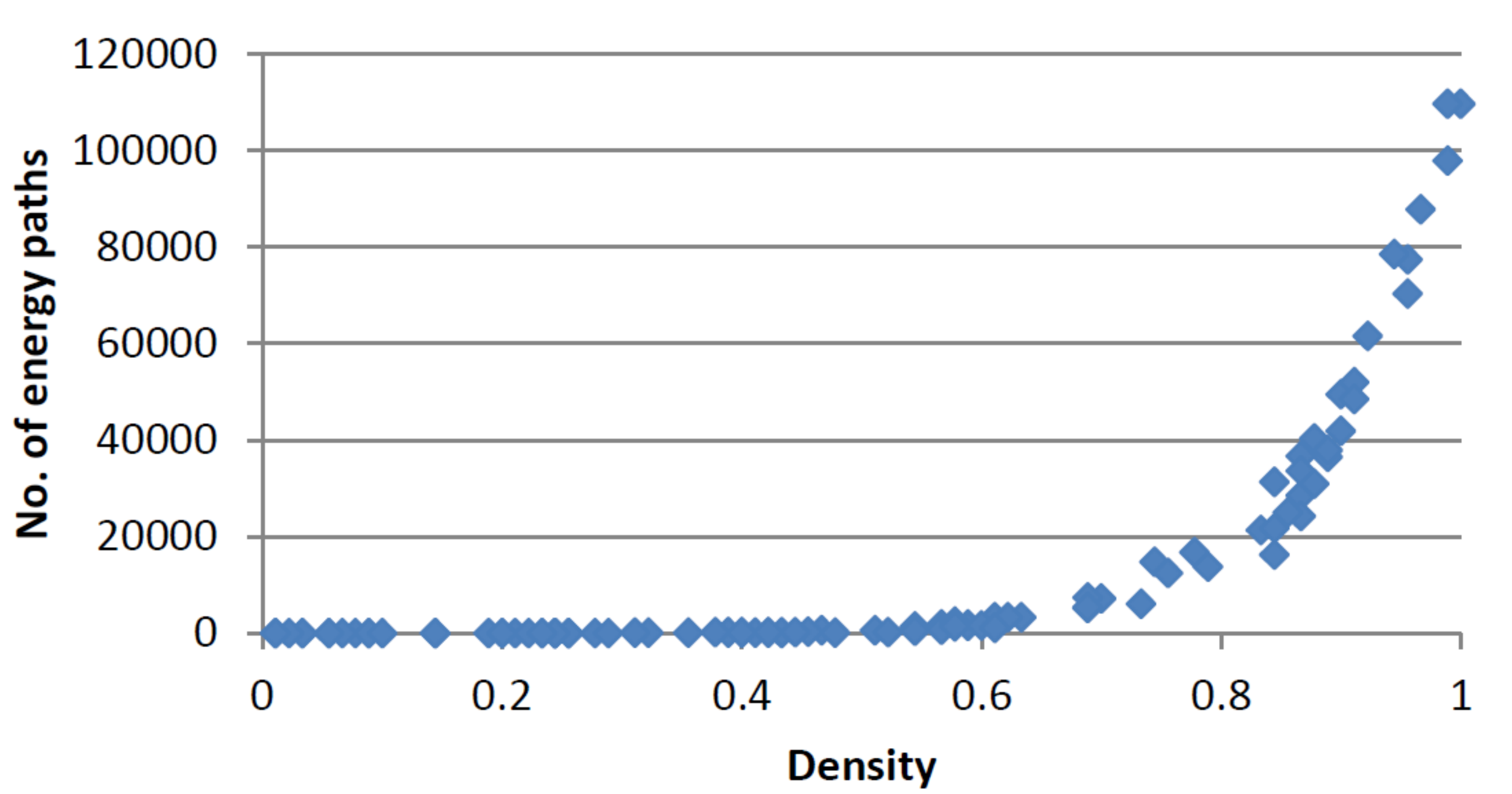}} 
	\end{center}
	\caption{Number of energy paths produced with different graph density and network size.}
  \label{fig:graphDensity}
\end{figure*}
Recall that $\tilde{G}(\tilde{\mathcal{N}},\tilde{\mathcal{A}})$ can reveal the accessibility of vehicular routes in $G(\mathcal{N},\mathcal{A})$ (see its definition in Section \ref{sec:construction}).
The longer the vehicular routes, the denser $\tilde{G}(\tilde{\mathcal{N}},\tilde{\mathcal{A}})$ . Hence we can associate the level of vehicular information disclosure to the network density of  $\tilde{G}(\tilde{\mathcal{N}},\tilde{\mathcal{A}})$, given by $\frac{|\tilde{\mathcal{A}}|}{|\tilde{\mathcal{N}}|(|\tilde{\mathcal{N}}|-1)}$. 
Since $|\mathcal{P}|$ grows very fast, we illustrate the growth with small networks only.
Fig. \ref{fig:graphDensity} shows the numbers of energy paths produced with different graph densities and network sizes. Figs. \ref{fig:rand4}--\ref{fig:rand10} correspond to networks with 4, 6, 8, and 10 nodes, respectively, each of which contains results of 100 random graphs with arbitrary source and destination pairs. We can see that $|\mathcal{P}|$ grows super-exponentially with network density. When we focus on a particular density, $|\mathcal{P}|$ grows  super-exponentially with network size as well. These confirm our analytical results related to $|\mathcal{P}|$ discussed in Lemma \ref{lm:cardP} and Theorem \ref{thm:runtime}. Therefore, when the network size is large and/or the lengths of vehicular routes are long, we generally cannot compute the whole $\mathcal{P}$ and thus Method I cannot be applied. Method II needs to be considered only when Method I is not applicable.

%


\subsection{Grid Network} \label{gridNetwork}
Next we compare the performance of the proposed methods on solving \eqref{minopt}. As Method I guarantees optimality, to evaluate the performance of Method III, we should compare Method III against Method I. To apply Method I, $\mathcal{P}$ should be manageable. We focus on a grid network of 16 road junctions, as shown in Fig. \ref{fig:grid}, where all the road connections are 10 km long with vehicles driven at 60 km/h. Suppose that there are 20 random vehicular routes, i.e., $|\mathcal{R}|=20$. We consider two cases: (i) each $r_i\in\mathcal{R}$ has an identical vehicular flow $f_i$ equal to 0.1 EVs per second;\footnote{0.1 EVs per second means that there are 0.1 participating EVs traversing the route in each second on the average.} (ii) each  $r_i$ has random flow $f_i\in[0.1,0.3]$ EVs per second. They represent different traffic conditions in a region with well-structured road network.

\begin{figure}[!t]
\centering
\includegraphics[width=1.5in]{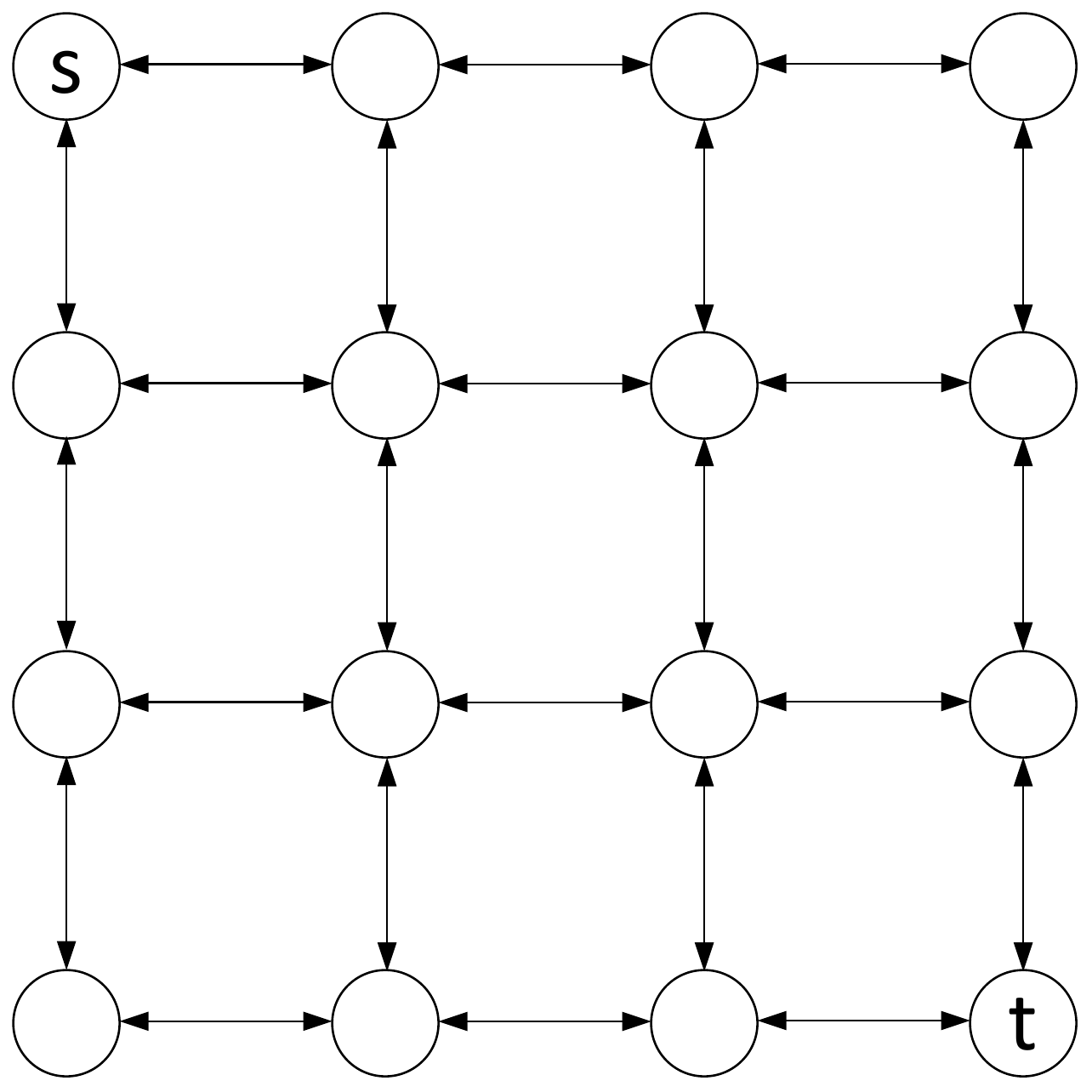}
\caption{A 16-Node Grid network.}
\label{fig:grid}
\end{figure}

Consider that we transmit energy from $s$ to $t$ indicated in Fig. \ref{fig:grid} with energy packet size of 1 kWh and 0.9 energy efficiency in a period of 5 hours, i.e., $w=1$ kWh, $z=0.9$, and $T=5$ hr. We solve \eqref{minopt} with a series of energy targets $\underline{X}$. Table \ref{case1} shows the total energy losses computed for Case (i). More energy transmission incurs more energy loss. We increase $\underline{X}$ from 1 to 1049 kWh, Methods I and III can produce exactly the same results. When $\underline{X}$ reaches 1050 kWh, the problem becomes infeasible. Table \ref{case2} shows results for Case (ii). The problem is feasible when $\underline{X}$ is smaller than or equal to 1962 kWh. For $\underline{X}\leq 900$ kWh, Method III can produce the optimal solutions. When $\underline{X}$ gets larger, Method III results in sub-optimal solutions with a little higher energy loss. When $\underline{X}$ is 1898 kWh or above, Method III can no longer produce feasible solutions. Despite this, Method III performs well when $\underline{X}$ is not too stringent. 
The performance of Method III degrades in higher $\underline{X}$ due to incorrect assignment of vehicular flow to energy paths. Method III establishes an energy path one at a time, followed by assigning its transmission rate. The path formation and rate assignment are done solely based on the properties of that energy path only. Since energy paths are inter-related, assigning vehicular flow to one energy path implies reducing the amount of available vehicular flow assignable to some other energy paths. When $\underline{X}$ is smaller, many ``good'' energy paths are available and the possibility of assigning vehicular flow to inappropriate energy paths is low. When $\underline{X}$ is higher, more energy paths are required and thus the possibility of assigning vehicular to flow to all the required energy paths is higher. When the flows of the vehicular routes vary to a greater extent, the possibility of misassignment of vehicular flow is higher. That is why Method III performs worse in Case (ii) than in Case (i). Method I is always superior in performance because it considers all possible energy paths when doing vehicular flow assignment.  
We have also tested other cases and the conclusions are similar. 

\begin{table}[!t]
\renewcommand{\arraystretch}{1.3}
\caption{Case {(i)}}
\label{case1}
\centering
\begin{tabular}{c|c|c|c}
\hline\hline
\multirow{2}{*}{$\underline{X}$}& \multicolumn{3}{c}{Total energy loss}\\ \cline{2-4}
 & Method I  & Method III & Difference (III-I)\\
\hline
1		&0.37			&0.37		&0\\
200	&74.35		&74.35	&0\\
400	&148.70		&148.70	&0\\
600	&223.05		&223.05	&0\\
800	&297.39		&297.39	&0\\
1000&	371.74	&371.74	&0\\
1010&	375.46	&375.46	&0\\
1020&	379.18	&379.18	&0\\
1030&	382.89	&382.89	&0\\
1040&	386.61	&386.61	&0\\
1049&	389.96	&389.96	&0\\
1050&	-				&	-			&	-\\
\hline\hline
\end{tabular}
\end{table}

\begin{table}[!t]
\renewcommand{\arraystretch}{1.3}
\caption{Case {(ii)}}
\label{case2}
\centering
\begin{tabular}{c|c|c|c}
\hline\hline
\multirow{2}{*}{$\underline{X}$}& \multicolumn{3}{c}{Total energy loss}\\ \cline{2-4}
 & Method I  & Method III & Difference (III-I)\\
\hline
1	&	0.52		&	0.52		&0\\
300	&	157.25	&	157.25	&0\\
600	&	314.49	&	314.49	&0\\
900	&	471.74	&	471.74	&0\\
1200&	633.61	&	650.16	&16.55\\
1500&	841.67	&	858.22	&16.55\\
1800&	1049.72	&	1066.27	&16.55\\
1897&	1116.99	&	1133.54	&16.55\\
1898&	1117.68	&	-		&-\\
1900&	1119.07	&	-		&-\\
1910&	1126.01	&	-		&-\\
1920&	1132.94	&	-		&-\\
1930&	1139.88	&	-		&-\\
1940&	1146.81	&	-		&-\\
1950&	1153.75	&	-		&-\\
1960&	1160.68	&	-		&-\\
1961&	1161.38	&	-		&-\\
1962&	1162.07	&	-		&-\\
1963&	-		&	-		&-\\
\hline\hline
\end{tabular}
\end{table}

\subsection{Real-world Scenario}

\begin{figure}[!t]
\centering
\includegraphics[width=2.7in]{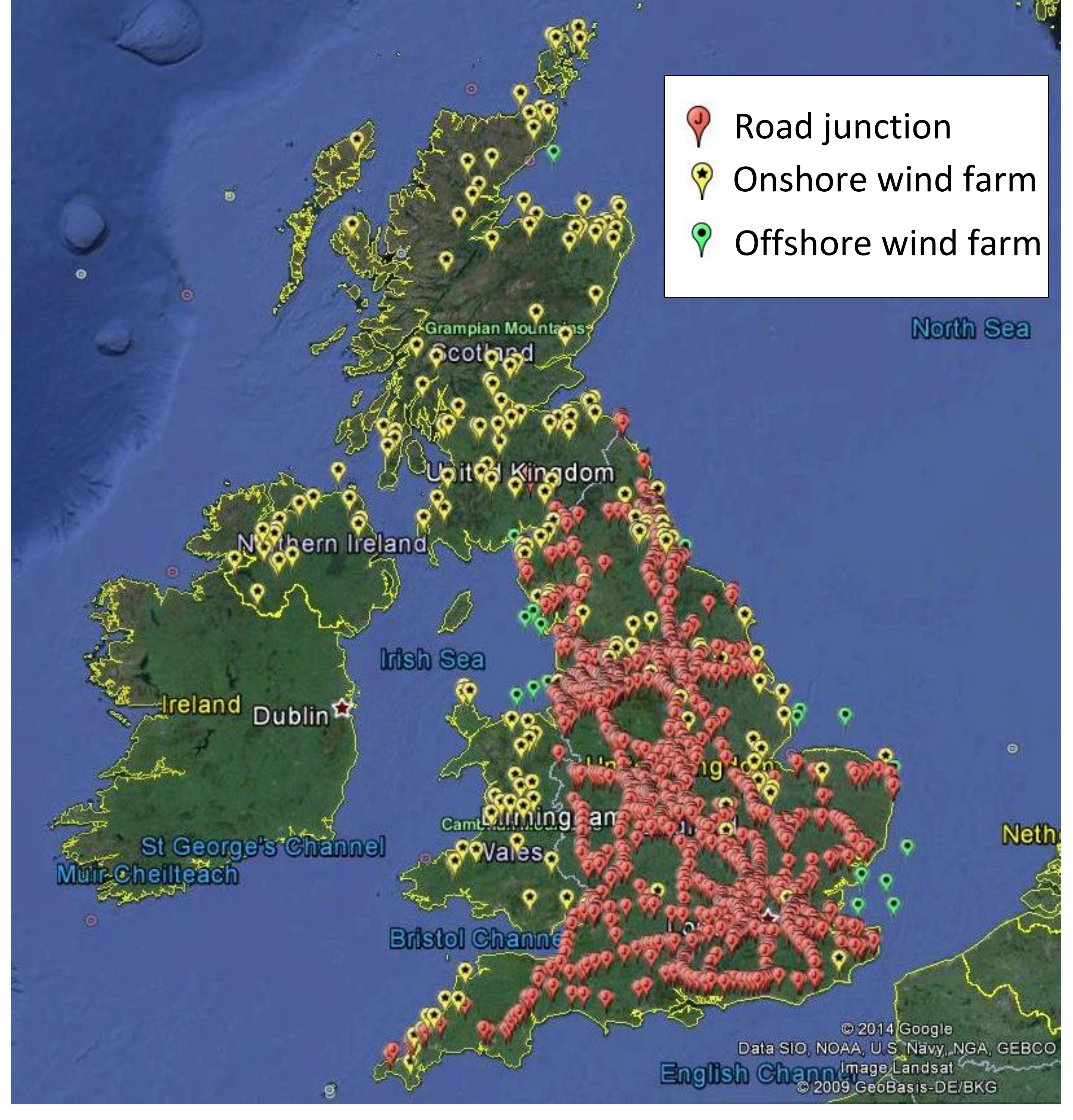} 
\caption{Locations of road junctions, and onshore and offshore wind farms \cite{VEN}.}
\label{fig:UKmap}
\end{figure}

Similar to \cite{VEN}, we study a VEN in a real-world setting. We build a VEN of 998 nodes and 2470 arcs based on  a highway network in the United Kingdom (U.K.). The road junctions and connections are configured according to the real-world data given in  \cite{ukroads}, where the locations of the road junctions are shown in Fig. \ref{fig:UKmap}. We select the traffic data of a date in June 2013 to set the travel times, vehicular speeds and flows. We randomly create 4788 vehicular routes, each of which  has a length no longer than 200 km.

U.K. has very rich wind-energy resources with annual production over $26\times 10^6$ MWh \cite{UKWindDB}. 203 onshore and 20 offshore wind farms have been built in the remote areas \cite{onshore,offshore} shown in Fig. \ref{fig:UKmap}, but there is   insufficient power infrastructure to bring the renewables online. Consider that we utilize the VEN to convey the renewables to urban areas with the objective of minimizing total energy loss. Suppose that the road junctions close to the wind farms and those located in London are the energy sources and destinations, respectively. We address \eqref{minopt} by selecting a source and a destination as $s$ and $t$. We set the energy target $\underline{X}$ to 10000 kWh and other settings are same as in  Section \ref{gridNetwork}.
Since the network is large, the whole set of $\mathcal{P}$ is not manageable and thus Method I is not applicable. Instead, we compare Methods II and III. For Method II, different random subsets of $\mathcal{P}$ are chosen for testing. Fig. \ref{fig:UKtest} illustrates the performance of Methods II and III, where each data point of Method II are the average of objective function values computed from 20 random subsets of $\mathcal{P}$.
With Method II, when the number of energy paths in the subsets of $\mathcal{P}$ increases, the total energy loss decreases. We can foresee that the total energy loss will converge to its optimal value when the number of energy paths selected approaches $|\mathcal{P}|$ (i.e., resulting in Method I). The computation time of Method II grows linearly with the number of energy paths adopted as most time is used to construct energy paths. When there are only a small number of energy paths selected, Method III can produce much better solutions than Method II and the computation time required is much lower. We can conclude that Method III is very effective at solving the power loss minimization problem.

\begin{figure}[!t]
\centering
\hspace{-1.0cm}
\includegraphics[width=3.8in]{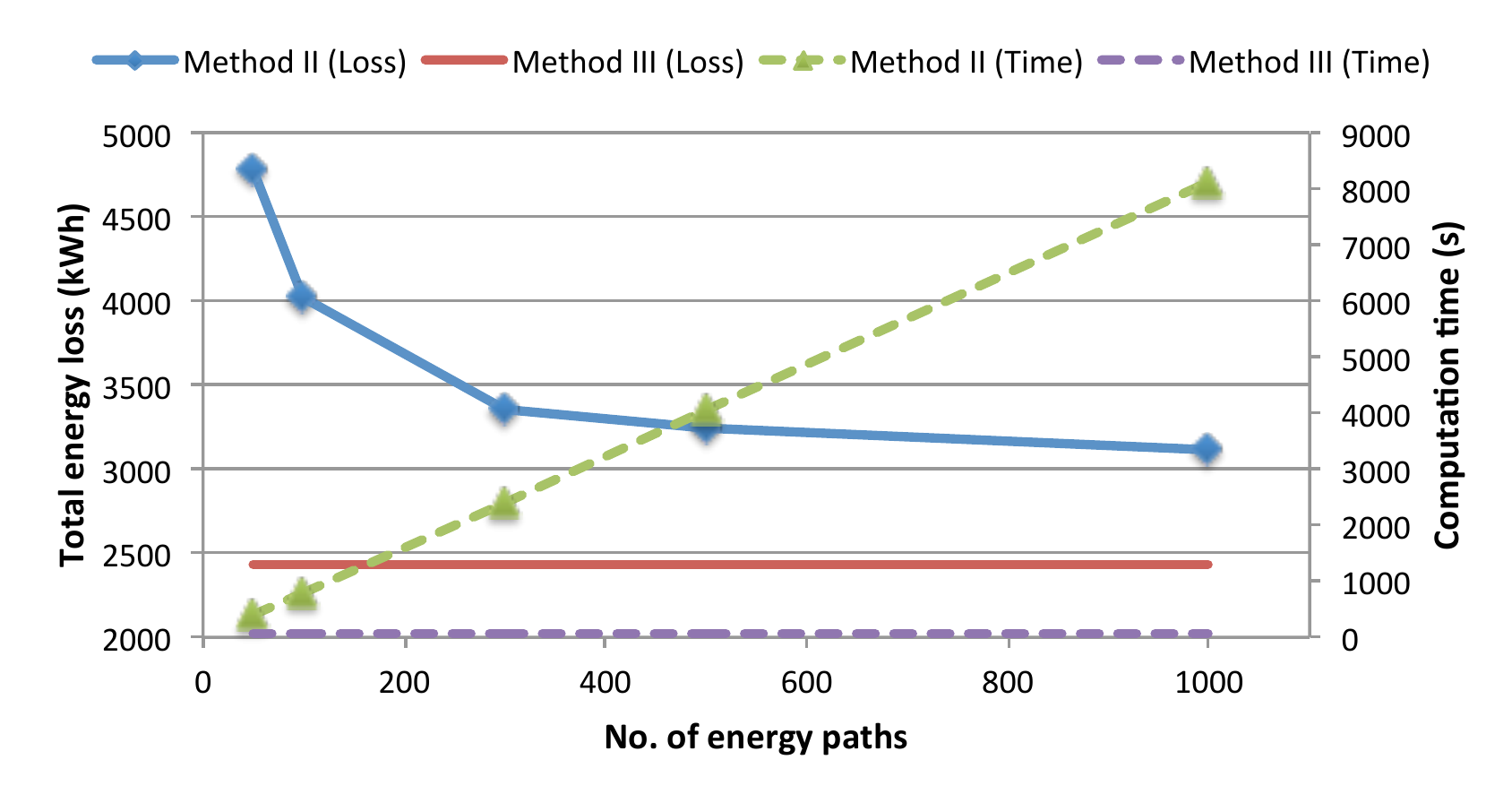}  \vspace{-1cm}
\caption{Performance of Methods II and III for the UK test case.}
\label{fig:UKtest}
\end{figure}

\subsection{Discussion} \label{sec:discussion}
We have proposed  three methods for VEN routing and they have different characteristics. We compare them in terms of four perspectives \textit{independently} in the following:

\subsubsection{Solution quality}
When $\mathcal{P}$ is manageable, Method I always performs the best as it can guarantee optimality. When addressing \eqref{minopt}, Method III outperforms Method II because the former can select energy paths out of all based on their properties while the latter depends on the quality of the selected subset of $\mathcal{P}$ given. Hence, we rank them as: I $>$ III $>$ II.

\subsubsection{Computation time}
The computation time required can be roughly measured by the number of energy paths examined. Since Method I needs to examine all energy paths, it takes the longest. Similarly, Method II examines the given subset of $\mathcal{P}$ only, it is faster than Method I. Method III only check those energy paths necessary to transmit energy but nothing more. Thus it requires the least amount of computation time.  Hence, they are ranked as: III $>$ II $>$ I.

\subsubsection{Solvable problem size}
$\mathcal{P}$  is not manageable when the network is large. Method I requires the whole $\mathcal{P}$, and thus, it cannot handle large problems. For Method II, the required number of energy paths as inputs is controllable and we can always input a manageable subset of energy paths based on the problem size. Method III only considers a sufficient number of energy paths which is not necessarily related to the problem size. So they can be ranked as: II $=$ III $>$ I.

\subsubsection{Applicability}
Since energy paths are the building blocks of VEN, as illustrated in \cite{VEN}, various VEN-related problems can be formulated in terms of optimization with $\mathcal{P}$. Methods I and II can be applied to these problems with minor modifications to the optimization formulation. However, Method III is a heuristic  tailor-made for the energy loss minimization problem only. Thus, we ranked them as: I $=$ II $>$ III.

\begin{table}[!t]
\renewcommand{\arraystretch}{1.1}
\caption{Solution method characteristic comparison}
\label{tab:comparison}
\centering
\begin{tabular}{p{2.5cm}|c | c|c}
\hline\hline
				& Method I 	& Method II & Method III 	\\
\hline
Solution quality & $\checkmark\checkmark\checkmark$		& $\checkmark$	& $\checkmark\checkmark$ \\
Computation time		& $\checkmark$	& $\checkmark\checkmark$ & $\checkmark\checkmark\checkmark$  \\
Solvable problem size &	$\checkmark$ 		& $\checkmark\checkmark$ 	& $\checkmark\checkmark$	\\
Applicability			 & $\checkmark\checkmark$ 		& $\checkmark\checkmark$ & $\checkmark$ \\
\hline\hline
\end{tabular}
\end{table}

We summarize their characteristics in Table  \ref{tab:comparison}. Methods I and III are good solution methods but  on two extremes: Method I is  comprehensive while Method III is focused. Method II is in the middle. Note that the performance of Method II can vary dramatically depending on the chosen subset of energy paths. For example, Method II can result in a good solution in short computation time so long as the chosen subset of energy paths is small but in good quality.
In the simulation above, the energy path subsets for Method II are randomly chosen but they are good enough to demonstrate the characteristics of the method. The problem of designing small energy path subsets with performance guarantee will be left as future work.


\section{Conclusion} \label{sec:conclusion}

VEN allows us to transport energy effectively across a large geographical area by means of EVs. With minimal modifications, a vehicular network can be easily converted into a VEN, which can function without changing the driving practice of the drivers. In a VEN, the energy sources and destinations are connected through a set of energy paths. This paper is dedicated to studying how to route energy over VEN by constructing energy paths opportunistically from a given set of vehicular routes. We give a method to construct all possible energy paths connecting a specified pair of energy source and destination, together with some analytical results for the method. This facilitates the determination of the optimal energy transmission schemes for various transmission objectives in the form of an LP. However, the construction of the whole set of energy paths usually requires super-exponential time. To reduce the computation time, we provide an alternative to formulate the optimization problems with a subset of energy paths,  but resulting in sub-optimality. For the power loss minimization problem, we develop a heuristic which is very efficient and capable of determining near-optimal solutions. We thoroughly test the performance of the three solution methods with artificial and real-world traffic networks. We also give a comprehensive comparison in terms of solution quality, computation time, solvable problem size, and applicability. This paper lays the foundations of VEN routing. In the future, we will improve the system performance by considering the properties of individual energy paths when only a limited number of energy paths are available. We will also extend this work for multi-source multi-destination routing.
\ifCLASSOPTIONcaptionsoff
  \newpage
\fi

\bibliographystyle{IEEEtran}
\bibliography{IEEEabrv}

\end{document}